\documentclass[journal,10pt]{IEEEtran}

\usepackage{amssymb,amsmath,amsthm}

\usepackage{url}

\usepackage{tikz}
\usetikzlibrary{calc}
\usetikzlibrary{positioning,circuits,circuits.logic,circuits.logic.US}
\usepackage{stackrel}

\usepackage{color}

\usepackage{graphicx}
\usepackage[english]{babel}
\usepackage{amsfonts}
\usepackage{booktabs}
\usepackage[noend]{algpseudocode}
\usepackage{multirow}
\usepackage{subcaption}
\usepackage{xifthen}
\usepackage{comment}
\usepackage{bbold}
\usepackage[inline]{enumitem}
\usepackage{orcidlink}
\usepackage{float}
\usepackage{framed}
\usepackage[operators]{cryptocode}

\newcommand{\Z}{\mathbb{Z}}
\newcommand{\Proj}{\text{Proj}}
\newcommand{\oracle}{\mathcal{O}}

\newcommand{\rand}{\text{Rand}}
\newcommand{\Garble}{\text{Garble}}
\newcommand{\Encode}{\text{Encode}}
\newcommand{\Eval}{\text{Eval}}
\newcommand{\Decode}{\text{Decode}}
\newcommand{\Gates}{\textup{\textsf{Gates}}}
\newcommand{\G}{\textsf{G}}
\newcommand{\Inputs}{\text{Inputs}}
\newcommand{\Outputs}{\text{Outputs}}
\newcommand{\Simulator}{\mathcal{S}}
\newcommand{\Hybrid}{\mathcal{G}}
\newcommand{\hash}{\mathcal{H}}
\newcommand{\new}[1]{\boxed{#1}}
\newcommand{\del}[1]{}

\newcommand{\adv}[2][]{\ifthenelse{\isempty{#1}}{\ensuremath{\text{Adv}(\text{#2})}}{\ensuremath{\text{Adv}_{#1}(\text{#2})}}}
\newcommand{\distinguisher}[1]{\ensuremath{\mathcal{D}_{\text{#1}}}}

\newcommand{\lsb}{\text{lsb}}
\newcommand{\length}{\ell}

\newcommand{\seck}{\kappa}

\renewcommand{\vec}[1]{\mathbf{#1}}

\newcommand{\scheme}{\Pi}

\theoremstyle{definition}
\newtheorem{definition}{Definition}
\theoremstyle{plain}
\newtheorem{theorem}{Theorem}

\newcommand{\codelocation}{\href{https://github.com/KULeuven-COSIC/gc-fast-sbox-eval}{https://github.com/KULeuven-COSIC/gc-fast-sbox-eval}}

\begin{document}
\title{Fast Evaluation of S-boxes with Garbled Circuits}
\author{Erik Pohle, Aysajan Abidin, Bart Preneel
        \thanks{COSIC KU Leuven, Belgium.}
        \thanks{This work is supported by CyberSecurity Research Flanders with reference number VR20192203 and by the Flemish Government through FWO SBO project MOZAIK S003321N.}
        }
\markboth{IEEE Transactions on Information Forensics and Security}{Fast Evaluation of S-boxes with Garbled Circuits}
\maketitle

\begin{abstract}
Garbling schemes are vital primitives for privacy-preserving protocols and secure two-party computation.
This paper presents a projective garbling scheme that assigns $2^n$ values to wires in a circuit comprising XOR and unary projection gates. A generalization of FreeXOR allows the XOR of wires with $2^n$ values to be very efficient. We then analyze the performance of our scheme by evaluating substitution-permutation ciphers. Using our proposal, we measure high-speed evaluation of the ciphers with a moderately increased cost in garbling and bandwidth. Theoretical analysis suggests that for evaluating the nine examined ciphers, one can expect a 4- to 70-fold improvement in evaluation performance with, at most, a 4-fold increase in garbling cost and, at most, an 8-fold increase in communication cost compared to the Half-Gates (Zahur, Rosulek and Evans; Eurocrypt'15) and ThreeHalves (Rosulek and Roy; Crypto'21) garbling schemes. In an offline/online setting, such as secure function evaluation as a service, the circuit garbling and communication to the evaluator can proceed in the offline phase. Thus, our scheme offers a fast online phase. Furthermore, we present efficient Boolean circuits for the S-boxes of TWINE and Midori64 ciphers. To our knowledge, our formulas give the smallest number of AND gates for the S-boxes of these two ciphers.
\end{abstract}

\section{Introduction}
\label{Introduction}
Privacy-preserving protocols enable collaborative computation on sensitive data while protecting the privacy of the sensitive data.
Successful implementations in a two-party scenario include privacy-preserving genome analysis~\cite{Jagadeesh2017}, email spam filtering~\cite{Gupta2017}, image processing~\cite{Chen2018} and machine learning~\cite{Kim2017}. The formalization of such two-party computation is called Secure Function Evaluation (SFE). Here the two parties, namely, Alice and Bob, want to compute a public function $f(x,y)$, where $x$ is the input of Alice and $y$ is the input of Bob, without revealing their input to each other.
Yao's garbled circuit protocol~\cite{Yao1986} has become a practical solution for SFE due to improved constructions. Moreover, garbling schemes (derived from the original garbled circuit construction) have also been identified as a useful cryptographic primitive. Most of the previous works focus on projective garbling schemes that assign two values to a wire, 0 and 1, such as the garbling scheme Half-Gates by Zahur et al.~\cite{Zahur2015} or the ThreeHalves scheme by Rosulek and Roy~\cite{Rosulek2021}.

This paper considers garbling schemes in the offline/online setting. The offline phase performs function-dependent pre-processing. Concretely, the garbler garbles the circuit computing $f$ and transmits the garbled gates to the evaluator but withholds the wire labels for the input layer. Once the input data of the garbler and the evaluator is available, the parties engage to obtain the appropriate wire labels for their respective inputs. Then, the evaluator evaluates the garbled circuit. The offline phase can be performed ahead of time and even batched to allow for optimal use of hardware and bandwidth if multiple function evaluations are expected. Hence, the online time, i.e., the time from having obtained the respective inputs to the evaluated output of the garbled circuit, is essential in this setting. This offline/online setting enables an efficient SFE as a service where the SFE service providers agree on a set of useful functions. The offline phase is run when the system is under low load and pre-processing results are stored. This way, the user of the service benefits from improved online times.

In this work, we examine a projective garbling scheme that assigns $2^n$ values to a wire. As a consequence, each wire in the circuit carries the semantics of an $n$-bit string. We generalize the encoding of FreeXOR by Kolesnikov and Schneider~\cite{Kolesnikov2008} to obtain a scheme where bitwise-XOR between $n$-bit strings is free. Our scheme allows fast evaluation of highly non-linear functions with $n$ input bits at a moderate additional garbling and bandwidth cost in the offline phase. We demonstrate this trade-off by implementing several symmetric-key primitives (SPN block ciphers with a cell structure as explained in Sect.~\ref{sec:applications}) and describe two application scenarios in Sect.~\ref{sec:applications:details}: distributed decryption in IoT-to-Cloud Secure Computation and distributing the key distribution center in the Kerberos authentication protocol. Garbling schemes for Boolean circuits need to express the non-linear S-box layer in the SPN primitives required in these two applications using multiple AND gates (e.g., for AES 32 AND gates) which translates to 64 and 96 hash function calls to evaluate each AES S-box in~\cite{Zahur2015} and \cite{Rosulek2021}, respectively. With our garbling scheme, evaluation of the entire S-box is done with one lookup table evaluation that costs one hash function call, yielding a speed-up by a large factor.

While the new garbling scheme assumes semi-honest adversaries, i.e. neither the garbler nor the evaluator may deviate from the protocol, several general approaches exist to make a garbled circuit protocol secure in the presence of active adversaries, which are allowed to deviate arbitrarily from the protocol. Prominent examples are based on cut-and-choose~\cite{Nielsen2009,Lindell2012,Huang2013,Huang2014}, on zero-knowledge proofs~\cite{Jarecki2007,Shelat2011} and authenticated garbling~\cite{Wang2017a,Katz2018,Dittmer2022}.
Moreover, semi-honest garbling schemes can be compiled into actively secure three party protocols in the honest majority setting~\cite{Mohassel2015}.

\subsubsection{Technical Overview}
The core ideas of the scheme are summarized as follows. We encode an $n$-bit string with bits $x_1, x_2, \dots x_n$ into a wire label as
\[
    W \oplus x_1R_1 \oplus x_2R_2 \oplus \dots \oplus x_nR_n \enspace,
\]
where $W$ is a random label. We call $R_i$ the wire label offsets that are randomly chosen by the garbler but fixed for all encodings in the circuit (see Definition~\ref{def:scheme:wire-label-offsets} for details) and if $x_i = 1$, $x_iR_i$ is $R_i$, otherwise it is the zero string. For $n=1$, this is the encoding of FreeXOR. Note that in previous garbling schemes with point-and-permute the wire labels have length $\approx \kappa$\footnote{To be precise: In, e.g., ~\cite{Zahur2015,Rosulek2021} the wire label length $k = \kappa+1$ since one bit is used as pointer bit.} while our encoding requires a length of $\kappa+n$ with $\kappa$ the security parameter.
We define two types of gates, XOR and projection gates. XOR gates compute the bitwise-XOR of two $n$-bit strings, require little garbling and evaluation work and are non-interactive, making them practically free.
A projection gate computes any $n$-bit to $m$-bit function on a wire value by using Yao's garbled table lookup, i.e., encrypting output wire labels using the respective input wire label as key. We apply standard garbled row reduction~\cite{Naor1999} and point-and-permute techniques~\cite{Beaver1990}. For a projection gate, the garbler's work is $2^n$ calls to the encryption primitive, and $2^n-1$ ciphertexts have to be sent to the evaluator, each ciphertext has size $\kappa+m$ bits. However, the evaluator only makes a single call to the encryption primitive, independent of the ``size'' $n$ of the projection gate. This makes the scheme attractive in the pre-processed garbled circuit model since any non-linear $n$-bit functionality can be evaluated with one call to the cryptographic primitive.

\subsubsection{Contributions}
We present a projective garbling scheme that assigns $n$-bit strings to each wire and in which XOR gates are free. The specific encoding of an $n$-bit string in a label allows seamless integration into existing garbling schemes that assign two values per wire. Following the spirit of modular proofs, we identify necessary properties of the cryptographic primitive (denoted by $\hash$) that is used to encrypt the truth table. Subsequently, we obtain a generalization of tweakable circular correlation robustness (TCCR, first defined by Choi et al.~\cite{Choi2012}), which we call $n$-TCCR, for $\hash$.

We apply the garbling scheme to compute nine symmetric-key primitives that follow a cell-based SPN architecture which is described in Sect.~\ref{sec:applications}, including AES, CRAFT, Fides, MANTIS, Midori, Piccolo, SKINNY, TWINE and WAGE. For these, we show a significant improvement in evaluation work in the online phase over the state-of-the-art schemes Half-Gates~\cite{Zahur2015} and ThreeHalves~\cite{Rosulek2021} that is traded off with moderate additional garbling work and/or communication cost in the offline phase. Table~\ref{tab:introduction:eval-improvement} shows the estimated evaluation improvement based on calls to $\hash$, which is complemented by a practical implementation in Sect.~\ref{sec:applications} that shows that this evaluation improvement translates into practice (see Table~\ref{tab:applications:spn:implementations}). Our code is publicly available\footnote{\codelocation}. We obtain evaluation times for, e.g., AES as low as 0.016 ms.

Furthermore, to facilitate implementation, we give Boolean circuits for the S-boxes of TWINE~\cite{Suzaki2012} and Midori64~\cite{Banik2015}, which is also used in MANTIS~\cite{Beierle2016} and CRAFT~\cite{Beierle2019}, using only AND and XOR gates. To the best of our knowledge, our Boolean circuits  give the smallest number of AND gates for these two ciphers, namely, 6 AND gates for TWINE's 4-bit S-box and 4 AND gates for the Midori64 $\textsf{Sb}_0$ S-box. Details can be found in Appendix~\ref{sec:appendix:sbox-formulas}.

\setlength{\tabcolsep}{0.5em}
\begin{table}
    \centering
    \caption{Evaluation work improvement for selected symmetric primitives over Half-Gates~\cite{Zahur2015}.
    Garbling and communication trade-off is listed in Table~\ref{tab:applications:spn:ciphers}.}
    \label{tab:introduction:eval-improvement}
    \resizebox{\linewidth}{!}{
    \begin{tabular}{lrr @{\hspace{2\tabcolsep}} lrr}
        \toprule
        Primitive & \multicolumn{2}{c}{Evaluation} & Primitive & \multicolumn{2}{c}{Evaluation}\\
        & \multicolumn{2}{c}{Improvement} & & \multicolumn{2}{c}{Improvement} \\
        \cmidrule{2-3}\cmidrule{5-6}
        \multicolumn{2}{r}{Half-Gates} &  \multicolumn{3}{c}{ThreeHalves \hfill Half-Gates} & ThreeHalves \\
        \midrule
        AES-128~\cite{NIST2001} & $26\times$ & $39\times$ & Piccolo-128 & $4.5\times$ & $6.8\times$ \\
        CRAFT~\cite{Beierle2019} & $5.7\times$ & $8.5\times$ & SKINNY-64-64~\cite{Beierle2016} & $7\times$ & $10\times$ \\
        Fides-80~\cite{Bilgin2013} & $15\times$ & $23\times$ & SKINNY-64-128 & $4.6\times$ & $7\times$ \\
        Fides-96 & $50\times$ & $75\times$ & TWINE-80~\cite{Suzaki2012} & $9.8\times$ & $14\times$ \\
        MANTIS~\cite{Beierle2016} &$4.3\times$ & $6.4\times$ & TWINE-128 & $9\times$ & $13\times$ \\
        Midori64~\cite{Banik2015} & $5.3\times$ & $8\times$ & WAGE~\cite{AlTawy2020} & $72\times$ & $109\times$ \\
        Piccolo-80~\cite{Shibutani2011} & $4.7\times$ & $7\times$ & & &\\
		\bottomrule
    \end{tabular}}
\end{table}

\subsubsection{Organisation}
The rest of the paper is organized as follows. In Sect.~\ref{sec:related-work}, we review related work on garbling schemes. Section~\ref{sec:background} gives more details on previous work which we build upon. We present our scheme in Sect.~\ref{sec:scheme} and prove its security in Sect.~\ref{sec:security}. A comparison of the state of the art and our scheme for nine (lightweight) symmetric primitives is given in Sect.~\ref{sec:applications}. Finally, we conclude the paper in Sect.~\ref{sec:conclusion}.
\section{Related Work}
\label{sec:related-work}
Recent improvements on Yao's garbled circuit protocol in the passive security setting focus on lowering bandwidth requirements, e.g.,~\cite{Pinkas2009,Kolesnikov2014}. In the line of work~\cite{Beaver1990,Naor1999,Kolesnikov2008} leading to the state-of-the-art schemes Half-Gates~\cite{Zahur2015} and ThreeHalves~\cite{Rosulek2021}, AND gates only require to send $2\seck$ bits and $\approx 1.5\seck$ bits, respectively, where $\seck$ is a security parameter, while XOR gates are free. Acharya et al.~\cite{Acharya2021} propose an approach to garbling where the garbled gate is no longer composed of ciphertexts from individual rows in the truth table, focusing only on binary gates.

While computation with binary values is mainly expressed in Boolean circuits with binary gates, gates with more inputs than two or more outputs than one have been studied as well. Dessouky et al.~\cite{Dessouky2017} define those gates as lookup tables and show how they can be evaluated in the passive security case in the Goldreich-Micali-Wigderson protocol~\cite{Goldreich1987}. FLUTE~\cite{Brueggemann2023} improves both setup and online cost of~\cite{Dessouky2017}. Damg{\aa}rd et al.~\cite{Damgaard2016,Damgaard2016a} design a table lookup for two-party secure computation and Keller et al.~\cite{Keller2017} extend it to the multi-party case based on secret-sharing. The basis for the aforementioned constructions is the one-time truth-table protocol OTTT by Ishai et al.~\cite{Ishai2013}. Table~\ref{tab:related-work:other-LUT-comparison} compares the (estimated) communication cost of these approaches for a 4-bit and 8-bit lookup table, respectively. Note that the aforementioned protocols require a function-dependent number of communication rounds in the online phase. While the local computations are faster than garbled circuit evaluation, the impact of network latency delaying the online computation time in this case is significantly increased by factor $\approx 10\times$ to $40\times$ for the SPN primitives (since they have 10 to 40 rounds). We aim to minimize online computation latency.

AES as a function has been studied explicitly by Durak and Guajardo~\cite{Durak2021}, SKINNY and Photon were studied by Abidin et al.~\cite{Abidin2023}. However, both works are in the arithmetic setting. 

\setlength{\tabcolsep}{.25em}
\begin{table}
    \centering
    \caption{Comparison of pre-processed lookup table (LUT) approaches in MPC protocols for $\kappa=128$. The depth of the circuit is denoted by $d$. Total communication is denoted in kilobytes (kB).}
    \label{tab:related-work:other-LUT-comparison}
    \begin{tabular}{llll}
        \toprule
        Scheme & Round  & \multicolumn{2}{c}{Total Comm. in kB} \\
        & Complexity & (4-bit LUT) & (8-bit LUT) \\
        \midrule
        OTTT~\cite{Ishai2013} (from \cite[Table IV]{Dessouky2017}) & $\mathcal{O}(d)$ & $\approx 6$ &  $\approx 262$ \\
        \midrule
        MiniMAC AES~\cite{Damgaard2016a} & $\mathcal{O}(d)$ & - & $\approx 700$ \\
        \midrule
        SP-LUT~\cite[Table IV]{Dessouky2017} & $\mathcal{O}(d)$ & $0.039$ & $0.288$ \\
        OP-LUT~\cite[Table IV]{Dessouky2017} & $\mathcal{O}(d)$ & $0.159$ & $\approx 65.5$ \\
        \midrule
        FLUTE~\cite[Table 2]{Brueggemann2023} & $\mathcal{O}(d)$ & \textbf{0.007} & \textbf{0.133} \\
        \midrule
        \midrule
        Fairplay~\cite{Malkhi2004}, Huang et el.~\cite{Huang2011} & $\mathcal{O}(1)$ & $1.024$ & $32.768$\\
        \midrule
        TASTY~\cite{Henecka2010} & $\mathcal{O}(1)$ & $0.96$ & $32.64$ \\
        \midrule
        Heath et al.~\cite{cryptoeprint:2024/369} & $\mathcal{O}(1)$ & $0.312$ & \textbf{1.392} \\
        \midrule
        This work & $\mathcal{O}(1)$ & \textbf{0.247} & $4.335$ \\
        \bottomrule
    \end{tabular}
\end{table}

In the garbled circuit domain, Fairplay~\cite{Malkhi2004} and TASTY~\cite{Henecka2010} already compute larger gates. Huang et al.~\cite{Huang2011} focus on an 8-bit to 8-bit AES S-box gate. 
Heath and Kolesnikov~\cite{Heath2021} construct a garbling gadget that computes a one-hot outer-product of two bit-vectors, which can be used to select one entry from a truth-table based on an index known by the evaluator. This approach has later been adapted to secret access to arbitrary truth-tables~\cite{cryptoeprint:2024/369}.
But unlike our scheme, these works consider multiple wires instead of multiple values \emph{per} wire. They also do not provide any security proof for the larger gates and Fairplay is vulnerable to attacks with malformed circuits~\cite{Nieminen2023}. Computing a gate with multiple input wires necessitates more generic hash function constructions that operate on inputs longer than one block. Practical performance improvements are due to the use of AES-NI instructions in permutation-based constructions such as~\cite{Bellare2013,Guo2020a,Chen2021}. It is unclear how these constructions extend to the multi-input case in the context of garbled circuits. Our scheme uses a cryptographic primitive with fixed-length input, enabling the use of AES-NI instructions. An overview of garbling and evaluation work, circuit size and the hash function construction for a generic $n$-to-$m$-bit gate is given in Table~\ref{tab:related-work:garbled-LUT-comparison}.

Lindell and Yanai~\cite{Lindell2018} investigate a projective garbling scheme with wires holding three-valued logic values but conclude that a translation into Boolean circuits is more efficient.

Since the work of Ball et al.~\cite{Ball2016} is conceptually very close to ours, we discuss it in detail in Sect.~\ref{sec:background:ball}. The main difference is that their scheme uses arithmetic circuits in $\Z_m$, where addition modulo an integer is free, while our proposal sticks to a bit representation where XOR is free.
This difference is essential for an efficient representation of SPN primitives that we target in this work. For $m > 2$ (otherwise it degenerates to a Boolean circuit for which~\cite{Zahur2015,Rosulek2021} offer better performance), XOR is a non-linear operation and thus would incur communication cost. Bitwise-XOR could either be emulated as $x+y-2xy$ where $x,y$ are $\Z_m$-encoded bits. This costs one multiplication gate per XOR, or $n$ multiplication gates to compute bitwise-XOR of $n$-bit strings, respectively. Alternatively, using projection gates in $\Z_m$ computing $\phi(x,y) = x \oplus y$ for $x,y \in \Z_{2^n}$ would require $2^{n+1}-1$ ciphertexts to be sent (disregarding the cost to compute $x||y$ as input first). Since in our scheme, XOR is linear, no ciphertexts need to be sent. We estimate that computing, e.g., SKINNY-64-128 in $\Z_m$ representation would be at least 3--4 times slower in all metrics compared to our scheme.

\begin{table}
    \centering
    \caption{Comparison of multi-input, multi-output gates in garbling schemes. We note the cost for a $n$-to-$m$-bit gate.\\
    $^\dagger$ $\kappa \approx 128-\max(n,m)$}
    \label{tab:related-work:garbled-LUT-comparison}
    \resizebox{\linewidth}{!}{
    \begin{tabular}{lllll}
        \toprule
        Scheme & Primitive & Garbling Work & Circuit Size & Evaluation Work \\
        & & & (in bits) & \\
        \midrule
        Fairplay~\cite{Malkhi2004} & SHA-1 & $2^n \cdot m \cdot 2$ & $2^n m\kappa$ & $m \cdot 2$\\
        \midrule
        TASTY~\cite{Henecka2010} & SHA-256 & $2^n \cdot m$ & $(2^n - 1) m\kappa$ & $m$ \\
        \midrule
        Huang et al.~\cite{Huang2011} & SHA-1 & $2^n \cdot \lceil \kappa m / 160 \rceil$ & $2^n m\kappa$ & $\lceil \kappa m / 160 \rceil$ \\
        \midrule
        Heath et al.~\cite{cryptoeprint:2024/369} & AES-NI & $\ge 2^n(1+m/\kappa)$ & $(n-1)\kappa$ & $\ge 2^n(1+m/\kappa))$ \\
        & & \quad $+nm$ & $+ 2^nm + nm$ & \quad$+nm$\\
        \midrule
        This work$^\dagger$ & AES-NI & $2^n$ & $(2^n - 1)\kappa$ & 1 \\
        & & & $+ (2^n-1)m$ & \\
        \bottomrule
    \end{tabular}}
\end{table}

\section{Background}
\label{sec:background}
We start with the arithmetic circuit scheme by Ball et al.~\cite{Ball2016} in Sect.~\ref{sec:background:ball} and detail the security model by Bellare, Hoang and Rogaway (BHR)~\cite{Bellare2012} in Sect.~\ref{sec:background:bhr}. Table~\ref{tab:background:notations} lists the notation used throughout the paper.

\setlength{\tabcolsep}{.25em}
\begin{table}\scriptsize
    \centering
    \caption{Notation.}
    \label{tab:background:notations}
    \begin{tabular}{|c|l|}
        \hline
        $\seck$ & Security parameter \\
        $\Vec{v}$ & Bold letters denote vectors \\
        $\{0,1\}^l$ & The set of bit-vectors of length $l$ \\
        $W_\alpha^\beta$ & The wire label of wire $\alpha$ that encodes the value $\beta$ \\
        $A \oplus B$ & Bitwise XOR for $A,B \in \{0,1\}^l$ \\
        $A||B$ & Bit-vector concatenation for $A \in \{0,1\}^l, B \in \{0,1\}^{l'}$ \\
        $\{A,B\}||C$ & Bit-vector concatenation extended to sets, i.e., $\{A||C, B||C\}$ \\
        $x' \gets x$ & Assignment of value $x$ to $x'$ \\
        $x \sample \{a,b,c,\dots\}$ & Uniform sampling from the set $\{a,b,c,\dots\}$ \\
        \hline
    \end{tabular}
\end{table}

\subsection{Garbled Circuits for Bounded Integers}
\label{sec:background:ball}
Ball et al.~\cite{Ball2016} propose a scheme based on garbled circuits that assigns integers $x \in \Z_m$ to each wire in the circuit. In this representation, addition (in $\Z_m$) is free in the same sense as FreeXOR. We briefly describe their scheme as our scheme is similar but represents $n$-bit strings per wire instead of numbers in $\Z_m$.

The wire encoding of $x \in \Z_m$ is $W_i^x = W_i^0 + x \odot \Delta_m \enspace,$
where $W_i^0, \Delta_m \in \Z_m^{\lambda_m}$, $\lambda_m = \lceil \frac{\seck}{\log_2 m} \rceil$. Addition is component-wise in the ring $\Z_m$. Here $\odot$ denotes a scalar multiplication. For each $m$, $\Delta_m$ is a secret, random vector known by the garbler. 

The scheme mainly offers two types of gates, addition and unary projection. For addition of wires $a$ and $b$ with output wire $c$, let $W_a^0, W_b^0$ be the two input wire labels of zero, then the garbler computes
$W_c^0 = W_a^0 + W_b^0$
as the output zero label. The evaluator, given $W_a^x$ and $W_b^y$ for evaluation, computes
\[
	W_c^{x+y} = W_a^x + W_b^y = \underbrace{W_a^0 + W_b^0}_{W_c^0} + (x+y) \odot \Delta_m \enspace.
\]
Addition incurs neither transmitted ciphertexts nor invocations of the encryption primitive.
Let $\phi: \Z_n \rightarrow \Z_m$ be an arbitrary function. The projection gate $\Proj_\phi$ computes the operation $x \mapsto \phi(x)$. Let $G$ be the garbled table, then the garbler fills $G[x + r]$ for every $x \in \Z_n$ as follows:
\[
    H(W_a^0 + x\odot\Delta_n) + W_c^0 + \phi(x) \odot \Delta_m = H(W_a^x) + W_c^{\phi(x)} \enspace,
\]
where $r$ is the secret cyclic shift offset. We can reduce the number of ciphertexts per projection gate to $n-1$ by applying garbled row reduction. The zero label is obtained when $r = -x$,
$W_c^0 = -H(W_a^{-r}) - \phi(-r) \odot \Delta_m$, from the encryption above. This, analogous to the binary case, fixes the ciphertext of the first \emph{garbled} row to $0^{\lambda_m}$.

Again, one element of the label can be used as a pointer and replace the shift $r$ during garbling if $\Delta_n$ is chosen appropriately. With this, the evaluator only has to decrypt the ciphertext the pointer indicates.

\subsection{Security Model by Bellare, Hoang and Rogaway}
\label{sec:background:bhr}
Bellare, Hoang and Rogaway~\cite{Bellare2012} define a security model for garbling schemes that formalizes the principle of circuit garbling as a cryptographic primitive. Many recent garbling schemes were proven secure in their model, e.g., \cite{Zahur2015,Kempka2016,Guo2020,Rosulek2021}. As we will use the same model, we give a brief overview.

A garbling scheme is a tuple of \Garble, \Encode, \Eval~and \Decode~algorithms:
\begin{itemize}
	\item \Garble: Transforms the input circuit $f$ into the tuple $(GC, e, d)$ where $GC$ is the garbled circuit, $e$ is the input encoding information (e.g., all semantic labels for input wires) and $d$ is the decoding information.
	\item \Encode: Encodes a given input $x$ using the semantic labels $e$ and returns a garbled input $X$, e.g., the input label with semantic $x$.
	\item \Eval: Evaluates the garbled circuit $GC$ using the input wire labels $\{W_i\}_{i \in \Inputs}$ and returns the output wire labels $\{W_i\}_{i \in \Outputs}$.
	\item \Decode: Decodes the output wire labels $\{W_i\}_{i \in \Outputs}$ using the decoding information $d$ and returns the plaintext output $y \in \{0,1\}^m$ or $\bot$ if the output wire labels are invalid.
\end{itemize}
The garbling scheme must produce correct circuit evaluations for any circuit $f$ and inputs $x \in \{0,1\}^n$.
Let $GC, e, d$ be the outputs of $\Garble(f)$, and $X_i$ the output of $\Encode(x_i, e)$ for $i \in \Inputs$ then $\Decode(\Eval(GC, \{ X_i \}_{i \in \Inputs}), d) = f(x)$
where $f(x)$ denotes the circuit evaluation in the clear.

Bellare et al. define two notions of secrecy. In the privacy notion, given $(GC, X, d)$, a party cannot learn any information besides what is revealed from the final output $y$ and the side-information function $\Phi$. In our case, $\Phi = \Phi_{\textup{topo*}}$ where only the circuit topology and the XOR gates are revealed but the function computed by projection gates remains hidden to the evaluator\footnote{In \cite{Bellare2012} $\Phi_{\textup{topo}}$ is defined as completely gate hiding, we therefore denote the slightly weaker notion with $\textup{topo*}$.}.
The privacy property can be achieved by giving a simulator $\Simulator$ for the \Garble~function that only receives the output $y$ and $\Phi$.
In the code-based game in Fig.~\ref{fig:background:prv.sim}, the garbling scheme is prv.sim secure if for every polynomial-time adversary $\mathcal{A}$ there is a polynomial-time simulator $\Simulator$ such that \adv{prv.sim}  is negligible, where
\[
    \resizebox{\linewidth}{!}{%
    $\adv{prv.sim} = \left|\Pr[\mathcal{A}\text{ wins prv.sim}] - \frac{1}{2}\right| = \left|\Pr[b = b'] - \frac{1}{2}\right|\,.$%
    }
\]
Intuitively, if the output of the simulator is indistinguishable from the output of \Garble~and \Encode~on a circuit and input chosen by the adversary, the scheme is prv.sim secure.
In the notion of obliviousness obv.sim, the adversary does not learn the decoding function. So given $(GC,X)$, a party cannot learn any information besides the side-information $\Phi$. 
The advantage is defined analogously to $\adv{prv.sim}$.

\begin{figure}[H]
\scriptsize
	\begin{subfigure}{0.49\columnwidth}
    \centering
    \algrenewcommand\algorithmicindent{0.5em}%
    \FrameSep1pt
    \begin{framed}
	\begin{algorithmic}
		\Function{\Garble}{$f,x$}
		\State $b \sample \{0,1\}$
		\If{b = 1}
			\State $(GC,e,d) \gets \Garble(f)$
			\State $X \gets \Encode(x, e)$
		\Else
			\State $y \gets f(x)$
			\State $(GC,X,d) \gets \Simulator(1^k, y, \Phi(f))$
		\EndIf
		\State \Return $(GC,X,d)$
		\EndFunction
	\end{algorithmic}
	\end{framed}
    \caption{Game prv.sim$_{\Phi, \Simulator}$.}
	\end{subfigure}
	\begin{subfigure}{0.49\columnwidth}
    \centering
    \algrenewcommand\algorithmicindent{0.5em}%
    \FrameSep1pt
	\begin{framed}
	\begin{algorithmic}
		\Function{\Garble}{$f,x$}
		\State $b \sample \{0,1\}$
		\If{b = 1}
			\State $(GC,e,d) \gets \Garble(f)$
			\State $X \gets \Encode(x, e)$
		\Else
            \State
			\State $(GC,X) \gets \Simulator(1^k, \Phi(f))$
		\EndIf
		\State \Return $(GC,X)$
		\EndFunction
	\end{algorithmic}
	\end{framed}
    \caption{Game obv.sim$_{\Phi, \Simulator}$.}
	\end{subfigure}
	\caption{For every circuit $f$ and input $x$ of the adversary's choice, the respective game function is called and the adversary outputs a choice $b'$ given $(GC,X,d)$ (resp. $(GC,X)$). The adversary wins if $b=b'$.}
	\label{fig:background:prv.sim}
\end{figure}
\section{The Scheme}
\label{sec:scheme}
In Sect.~\ref{sec:scheme:circuit-definition}, we first describe the notation for a circuit comprising XOR gates and projection gates. Then, we detail how the garbler encodes $n$-bit strings and transforms them into wire labels. Next, in Sect.~\ref{sec:scheme:gates}, we show how XOR gates are garbled and evaluated, followed by a description of how projection gates are garbled and evaluated.
Section~\ref{sec:scheme:circuit-constructs} describes higher-level gadgets that can be obtained from the aforementioned gates. In Sect.~\ref{sec:scheme:garbling-scheme}, all concepts are pieced together to describe the garbling, evaluation and decoding function. We also describe how input is handled. The complete garbling scheme $\scheme$ is given in Fig.~\ref{fig:scheme:scheme}.
We start with some general notations.
Let $\lsb_n(W)$ be the $n$ least significant bits\footnote{The exact location of the $n$ bits in $W$ is not important for the scheme as long as it is consistently used by both parties.} of the bit-vector $W \in \{0,1\}^k$. With $k$ we denote the wire label length.
We use a hash function $\hash: \{0,1\}^k \times \{0,1\}^\tau \rightarrow \{0,1\}^k$ that accepts a $k$-bit input, a $\tau$-bit tweak and outputs $k$ bits. Further properties of $\hash$ are presented in Sect.~\ref{sec:security:ntccr-def}.

\subsection{Circuit Definition}
\label{sec:scheme:circuit-definition}
We define a circuit with a $p$-bit input and $q$ gates. The function computed by the circuit is denoted by $f$. 
Let the wire index be $1,\ldots,p,p+1,\ldots,p+q$, where the input wires have index $1,\ldots,p$ and the output wire of the $i$-th gate has index $p+i$. 
We denote the set of input wire indices as $\Inputs$, and the set of output wire indices as $\Outputs$. 
We associate a bit-length $\length(i)$ to each wire $i$.
Let $\bar{n}$ denote the maximum bit-length of wires used in $f$, then we use bit strings of length $k = \seck + \bar{n}$ as wire labels.
Let \Gates~be a topologically sorted list of gates $\G_1, \ldots, \G_q$. 
We distinguish two types of gates: XOR and projection gates. XOR gates accept two wires of the same bit-length $n$ as input and output a wire with bit-length $n$. The unary projection gate accepts one $n$-bit wire and outputs one $m$-bit wire.
\begin{figure}
	\centering
	\begin{framed}
    \small
    \raggedright Parameters:\\
    \quad $\bar{n}$ maximum bit-length of wires in $f$,\\
    \quad $k$ wire label length with $k = \seck + \bar{n}$
    \algrenewcommand\algorithmicindent{0.5em}%
		\begin{algorithmic}[1]
			\Function{GenR}{$\bar{n}$}
    			\For{$i \in \{1,\ldots,\bar{n}\}$}
        			\State $\vec{R}_i \sample \mathcal{R}_i$ \Comment{cf. Definition~\ref{def:scheme:wire-label-offsets}}
    			\EndFor
    			\State \Return $\vec{R}_1, \ldots, \vec{R}_{\bar{n}}$
			\EndFunction
			\Statex 
			
			\Function{\Garble}{$f$}
    			\State $\vec{R}_1, \ldots, \vec{R}_{\bar{n}} \gets \Call{GenR}{\bar{n}}$
    			\For{$i \in \Inputs$}
    				\State $W_i^{0^{\length(i)}} \sample \{0,1\}^k$
    			\EndFor
    			\For{$\G_i \in \Gates$}
    				\If{$\G_i = XOR$}
    					\Comment $\G_i$ with $n$-bit input wires $a$,$b$
    					\State $W_i^{0^n} \gets W_a^{0^n} \oplus W_b^{0^n}$
    				\Else
    				\Comment $\G_i$ with $n$-bit input wire $a$
    				\State $W_i^{0^m} \sample \{0,1\}^k$ \Comment and $\phi: \{0,1\}^n \rightarrow \{0,1\}^m$
    				\For{$x \in  \{0,1\}^n$}
    				\State $GC[i,\lsb_n(W_a^x)] \gets \hash(W_a^x, i) \oplus W_i^{0^m} \oplus \phi(x) \cdot \vec{R}_m$
    				\EndFor
    				\EndIf
    			\EndFor
                \State $d \gets \{ \lsb_{\length(i)}(W_i^{0^{\length(i)}}) \}_{i \in \Outputs}$
    			\State $e \gets \{W_i^{0^{\length(i)}} \}_{i \in \Inputs}, \vec{R}_1, \ldots, \vec{R}_{\bar{n}}$
    			\State \Return $GC, e, d$
			\EndFunction
			
			\Statex 
			\Function{\Encode}{$e,\{x_i\}_{i \in \Inputs}$}
			    \For{$i \in \Inputs$}
			    \State $X_i \gets W_i^{0^{\length(i)}} \oplus x_i \cdot \vec{R}_{\length(i)}$
			    \EndFor
			    \State \Return $\{X_i\}_{i \in \Inputs}$
			\EndFunction
			
			\Statex 
			\Function{\Eval}{$GC,\{W_i\}_{i \in \Inputs}$}
    			\For{$\G_i \in \Gates$}
        			\If{$\G_i = XOR$}
            			\Comment $\G_i$ with input wires $a$, $b$
            			\State $W_i \gets W_a \oplus W_b$
        			\Else
            			\Comment $\G_i$ with $n$-bit input wire $a$
            			\State $W_i \gets \hash(W_a,i) \oplus GC[i,\lsb_n(W_a)]$
        			\EndIf
    			\EndFor
    			\State \Return $\{W_i\}_{i \in \Outputs}$
			\EndFunction
			\Statex 
			\Function{\Decode}{$\{W_i\}_{i \in \Outputs}$}
    			\For{$i \in \Outputs$}
    			    \State $y_i \gets d_i \oplus \lsb_{\length(i)}(W_i)$
    			\EndFor
    			\State \Return $y$
			\EndFunction
		\end{algorithmic}
	\end{framed}
	\caption{The new garbling scheme $\scheme$ comprises a garble, evaluation, encoding and decoding function.}
	\label{fig:scheme:scheme}
\end{figure}

\begin{definition}[Wire Label Offsets]
	For each bit-length $n$ ($1 \le n \le \bar{n}$) that is used in $f$, a wire label offset is a bit-vector of length $k = \seck + n$ with $\seck$ random bits and $n$ fixed bits. 
	The garbler draws the matrix $\vec{M}$ uniformly at random from $\{0,1\}^{\kappa \times n}$ and appends fixed bits to each column-vector to form $\vec{R}_n
		= \begin{pmatrix}\vec{M}\\\vec{I}_n\end{pmatrix}\,,
	$
	where $\vec{I}_n \in \{0,1\}^{n \times n}$ is the identity matrix. The column vector $R_i$ in $\vec{R}_n$ is the $i$-th wire label offset. We denote the distribution from which $\vec{R}_n$ is sampled $\mathcal{R}_n$, i.e., $\vec{R}_n \sample \mathcal{R}_n$.
	\label{def:scheme:wire-label-offsets}
\end{definition}
The matrix $\vec{R}_n$ is used throughout the whole circuit for all wires of bit-length $n$. We use the last $n$ bits of the label to fix distinct values to allow point-and-permute~\cite{Beaver1990}.
The inner product of $x \cdot \vec{R}_n$ is defined as $x_1 R_1 \oplus \ldots \oplus x_n R_n$. 

\begin{definition}[Wire Label Encoding]
	The encoding $W_i^x$ of an $n$-bit string $x \in \{0,1\}^n$ on a wire with index $i$ is defined as
    $W_i^x = W_i^{0^n} \oplus x \cdot \vec{R}_n$.
	\label{def:scheme:wire-label-encoding}
\end{definition}
Note, this yields a unique encoding for all $x$ and $\vec{R}$ even if the random part $\vec{M}$ is linearly dependent in the columns because the lower $n$ bits of $x \cdot \vec{R}$ are always unique due to $\vec{I}_n$.

Intuitively, there are $n$ distinct offsets $R$, one for each encoded bit. The offset applied to a wire label that encodes $x$ is the linear combination of $R$ values. 

\subsection{Gates}\label{sec:scheme:gates}
For an XOR gate with $n$-bit input wires $a$ and $b$, and output wire $c$, the garbler generates the output wire label $W_c^x \gets W_a^{0^n} \oplus W_b^{0^n} \oplus~x\cdot\vec{R}_n$ where $x \in \{0,1\}^n$.
No ciphertext is sent.

To evaluate an XOR gate, let $W_a$ and $W_b$ be the wire labels that the evaluator obtained as input labels for the XOR gate. The output label is then computed as $W_c \gets W_a \oplus W_b$.

A projection gate $\Proj_\phi$ computes the unary projection $\phi: \{0,1\}^n \rightarrow \{0,1\}^m$, a $n$-to-$m$-bit function.
Let $a$ be the input wire index to the projection gate and $c$ be the index of the output wire, the garbler first draws the output wire label for 0 at random: $W_c^{0^m} \sample \{0,1\}^k$ and then generates $2^n$ ciphertexts for each $x \in \{0,1\}^n$ and stores the result in the garbled table at the position indicated by the pointer bits, i.e., $GC[c,\lsb_n(W_a^x)] \gets \hash(W_a^x,c) \oplus W_c^{\phi(x)}$.
We apply the row-reduction technique~\cite{Naor1999} and reduce the number of ciphertexts that need to be sent by one.
Let $a$ be the input wire index to the projection gate $\Proj_\phi$ and $c$ be the index of the output wire.
Then, the garbler chooses the output wire label for $0^m$ as
\[
	W_c^{0^m} = \hash(W_a^{\lsb_n(W_a^{0^n})}, c) \oplus \phi(\lsb_n(W_a^{0^n})) \cdot \vec{R}_m
\]
and computes the remaining ciphertexts as described above.
Since the first ciphertext (where $x=\lsb_n(W_a^{0^n})$) is always $0^k$, it does not need to be sent. The number of rows sent to the evaluator is therefore $2^n - 1$.

For evaluation, let $W_a$ be the wire label that the evaluator obtained as input to the projection gate. The output label $W_c$ is computed by
\[
	W_c \gets GC[c,\lsb_n(W_a)] \oplus \hash(W_a,c)~,
\]
where the position of the ciphertext to evaluate is indicated by the pointer bits of the input wire label. The first ciphertext is set to $0$: $GC[c,0^n] = 0^m$.

\subsection{Circuit Constructions}
\label{sec:scheme:circuit-constructs}
Below, we give useful gadgets comprised of XOR and projection gates.
\begin{description}
\item[Wire Composition.]
We can compose an $n$-bit wire $a$ with an $m$-bit wire $b$ resulting in a $(n+m)$-bit wire $c$. The composition construction computes the functionality $f: \{0,1\}^n \times \{0,1\}^m \rightarrow \{0,1\}^{n+m}$ defined by $f(x,y)= x||y$.
The composition is then $W_c \gets \Proj_s(W_a) \oplus \Proj_{s'}(W_b)$,
where $s: \{0,1\}^n \rightarrow \{0,1\}^{n+m}$ is defined as $s(x) = x||0^m$ and $s': \{0,1\}^m \rightarrow \{0,1\}^{n+m}$ is given as $s'(y) = 0^n||y$. 
Wire composition costs $2^n + 2^m$ ciphertexts to garble and two ciphertexts to evaluate.

Note that the construction is not limited to two arguments. It is efficient to compose many wires together at once instead of cascading or using a tree-based approach\footnote{Following the example, the tree-based approach first composes $a||b$ and $c||d$ resulting in 2-bit wires. Then $ab||cd$ is composed.}. E.g., to compose four $1$-bit wires $a,b,c,d$ , we may use
\[
\Proj_{s_a}(W_a) \oplus \Proj_{s_b}(W_b) \oplus \Proj_{s_c}(W_c) \oplus \Proj_{s_d}(W_d)\enspace,
\]
where $s_a(x) = 000||x$, $s_b(x) = 00||x||0$, $s_c(x) = 0||x||00$, $s_d = x||000$. This costs $4\cdot 2^1 = 8$ ciphertexts (or 4 with row reduction) instead of $4\cdot 2^1 + 2 \cdot 2^2 = 16$ (resp. 10) ciphertexts.

\item[Wire Decomposition.]
Likewise, we can decompose, i.e., split, a $2n$-bit wire into two $n$-bit wires.
Let $W_a$ be a $2n$-bit wire, then the decomposition construction computes $f: \{0,1\}^{2n} \rightarrow \{0,1\}^{n}$ defined as $f(x_1||\ldots||x_{2n}) = x_1||\ldots||x_n$ and $f': \{0,1\}^{2n} \rightarrow \{0,1\}^{n}$ as $f'(x_1||\ldots||x_{2n}) = x_{n+1}||\ldots||x_{2n}$ via two projection gates. 
Note that this time, a tree-like decomposition, e.g., from 4-bit to 2-bit to 1-bit, is more efficient than constructing four projections from 4-bit to 1-bit. The latter costs $4 \cdot 2^4 = 64$ ciphertexts (60 with row reduction) while the former costs $2 \cdot 2^4 + 4 \cdot 2^2 = 48$ (resp. 42) ciphertexts.

\item[Constants.]
In the garbling scheme, we can encode public constants or constants known only to the garbler at no cost.
Let $x \in \{0,1\}^n$ be the constant for the $n$-bit wire $a$, then the garbler chooses $W_a^{0^n} \gets x \cdot \vec{R}_n$.
This fixes the label $W_a^x$ to $0^k$. No ciphertext is sent to the evaluator. Likewise, the evaluator uses $W_a = 0^k$ for further evaluation.
\end{description}

\subsection{Garbling Scheme}
\label{sec:scheme:garbling-scheme}
We now describe the complete garbling scheme (see Fig.~\ref{fig:scheme:scheme}).
\begin{description}
\item[Garble.]
The garbler chooses $\bar{n}$ matrices of offset values (see Definition~\ref{def:scheme:wire-label-offsets}). For each input bit $i$, a wire label $W_i^0$ is chosen uniformly at random. The garbling process applies the operations for projection and XOR gates as described in Sect.~\ref{sec:scheme:gates} gate-by-gate in topological order. In the end, the garbling routine outputs the ciphertexts, input wire values, offsets and decoding information.

\item[Encoding and Oblivious Transfer.]
\label{sec:scheme:ot}
The garbler encodes their own input by picking the respective wire label.
In Yao's protocol, the evaluator obtains the appropriate wire labels that correspond to its input via oblivious transfer (OT)~\cite{Rabin2005}. Using OT extensions~\cite{Ishai2013,Asharov2013} speeds this up in practice. To obtain the correct label for an $n$-bit wire, one could simply perform a 1-out-of-$2^n$ OT. Naor and Pinkas~\cite{Naor1999a} show how to reduce this to $n$ 1-out-of-2 OTs by introducing additional
pseudorandom function  (PRF) evaluations. However, using the FreeXOR property of our scheme, we can instead perform only $n$ 1-out-of-2 OTs (as in a garbling scheme with 2 wire labels). For each input bit $b_i$ at position $i$, the sender sends
\begin{equation*}
    \begin{array}{lll}
        W_i^{0^n} & & \text{ if $b_i = 0$}\,, \\
	    W_i^{0^n} & \oplus~R_i & \text{ if $b_i = 1$}\,,
    \end{array}
\end{equation*}
where $R_i$ is the $i$-th column vector in $\vec{R}_n$. To obtain the wire label for the $n$-bit wire, we XOR the obtained labels together at no additional cost.  Note that $W_i^{0^n}$ is a fresh random wire label for each bit $i$ of the input.

\item[Evaluation and Decoding.]
Once the evaluator obtains the garbled inputs, it computes the garbled output of each gate accordingly (see Sect.~\ref{sec:scheme:gates}). 
Having computed the garbled output, the evaluator may either share the wire labels with the garbler or directly use the decoding information $d_i = \lsb_n(W_i^{0^n})$ for output wire $i \in \Outputs$ in the decoding function to obtain the output bits in the clear. 
Let us briefly look at why this decoding scheme is correct. Let $i$ be an output wire. Since we fixed $\lsb_n(y \cdot \vec{R}_n) = y \cdot \vec{I}_n = y$ by construction of the offset values, for any value $y \in  \{0,1\}^n$, we have $\lsb_n(W_i^y) = \lsb_n(W_i^{0^n}) \oplus y$. 
As $d_i = \lsb_n(W_i^{0^n})$, the decoding is correct
\[\resizebox{\linewidth}{!}{%
$d_i \oplus \lsb_n(W_i^x) = \lsb_n(W_i^{0^n}) \oplus \lsb_n(W_i^{0^n}) \oplus \lsb_n(y\cdot\vec{R}_n) = y\,.$}
\]
\end{description}
\section{Security}
\label{sec:security}
Using the BHR security model (see Sect.~\ref{sec:background:bhr}) we show that if a hash function satisfies the properties of $n$-TCCR security defined in Sect.~\ref{sec:security:ntccr-def} below, our scheme is prv.sim (Sect.~\ref{sec:security:prv.sim}) and obv.sim (Sect.~\ref{sec:security:obv.sim}) secure. We sketch how to achieve authenticity in Sect.~\ref{sec:security:authenticity}.

\subsection{(n-)TCCR Security}
\label{sec:security:ntccr-def}
We revisit the tweakable circular correlation robustness (TCCR) definition by Guo et al.~\cite{Guo2020a} adapted to our notation.
\begin{definition}[TCCR Security~\cite{Guo2020a}]
	A TCCR (tweakable circular correlation robust) hash function $\hash$ is a function $\{0,1\}^k \times \{0,1\}^\tau \rightarrow \{0,1\}^k$ that accepts a message $m$ and a tweak $t$. In the TCCR security game, the distinguisher \distinguisher{TCCR} is given one of the two oracles with signature $\{0,1\}^k \times \{0,1\}^\tau \times \{0,1\} \rightarrow \{0,1\}^k$
	\begin{itemize}
		\item (Real) $\oracle_R(m,t,b) = \hash(m \oplus R, t) \oplus bR$
		\item (Ideal) $\rand(m,t,b)$ is a random function.
	\end{itemize}
	with the goal to decide which is the oracle given to it. The distinguisher doesn't know the secret value $R \in \{0,1\}^k,\, R \sample \mathcal{R}_{\textup{TCCR}}$ and is only allowed to make legal queries. An illegal query is $(m,t,1-b)$ if $(m,t,b)$ has been queried before.
	
	\noindent We define the advantage as
	\[\begin{aligned}
		&\adv[\mathcal{R}_{\textup{TCCR}}]{\distinguisher{TCCR}} \\
            &\quad= \left| \Pr[\distinguisher{TCCR}^\rand(1^\seck) = 0] - \Pr_{R \gets \mathcal{R}_{\textup{TCCR}}}[\distinguisher{TCCR}^{\oracle_R}(1^\seck) = 0] \right|\,,
        \end{aligned}
	\]
	where $\distinguisher{}^\oracle$ signifies that the distinguisher has access to oracle $\oracle$. We call $\hash$ TCCR secure if $\adv[\mathcal{R}]{\distinguisher{TCCR}}$ is negligible in the security parameter $\seck$.
	\label{def:tccr}
\end{definition}
Note that the advantage of $\distinguisher{TCCR}$ depends on the distribution $\mathcal{R}_{\textup{TCCR}}$ of the secret value $R$.
Next, we define $n$-TCCR security, a generalized TCCR notion incorporating $n$ secret offsets.
\begin{definition}[n-TCCR Security]\label{def:n-tccr}
	A n-TCCR hash function $\hash$ is a function $\{0,1\}^k \times \{0,1\}^\tau \rightarrow \{0,1\}^k$ that accepts a message $m$ and a tweak $t$.
	In the n-TCCR security game, the distinguisher \distinguisher{n-TCCR} is given one of the two oracles with signature $\{0,1\}^k \times \{0,1\}^\tau \times \{0,1\}^n \times \{0,1\}^n \rightarrow \{0,1\}^k$
	\begin{itemize}
		\item (Real) $\oracle_{\vec{R}}(m,t,\vec{a}, \vec{b}) = \hash(m \oplus \vec{a}\cdot\vec{R}, t) \oplus \vec{b}\cdot\vec{R}$
		\item (Ideal) $\rand(m,t,\vec{a},\vec{b})$ is a random function
	\end{itemize}
	with the goal to decide which is the oracle given to it. 
	We interpret $\vec{a}, \vec{b} \in \{0,1\}^n$ as binary vectors, $\vec{R} = (R_1, \dots, R_n),\, \vec{R} \in \{0,1\}^{k \times n} \sample \mathcal{R}_n$ and $R_i \in \{0,1\}^k$, $1 \le i \le n$. The expression $\vec{a}\cdot\vec{R} = a_1 R_1 \oplus \dots \oplus a_n R_n$ is the linear combination of offsets defined by $\vec{a}$.
	The distinguisher doesn't know the secret value $\vec{R}$ and is only allowed to make legal queries. An illegal query is $\vec{a} = 0$ or $(m,t,\vec{a}, \vec{b}')$ if $(m,t,\vec{a},\vec{b})$ has been queried before for $\vec{b} \neq \vec{b}'$.
	
	\noindent We define the advantage as
	\[\begin{aligned}
		&\adv[\mathcal{R}_n]{\distinguisher{n-TCCR}} \\
            &\quad= \left| \Pr[\distinguisher{n-TCCR}^{\textsf{Rand}}(1^\kappa) = 0] - \Pr_{\vec{R} \gets \mathcal{R}_n}[\distinguisher{n-TCCR}^{\mathcal{O}_{\vec{R}}}(1^\kappa) = 0] \right|\,.
        \end{aligned}
	\]
	We call $\hash$ $n$-TCCR secure if $\adv[\mathcal{R}_n]{\distinguisher{n-TCCR}}$ is negligible in $\seck$.
\end{definition}

Clearly, every $n$-TCCR secure hash function is also TCCR secure. Since in $n$-TCCR, the distinguisher has more freedom regarding queries to the oracle, a statement about the inverse direction is not straightforward.

Limitations and assumptions on TCCR are realized once instantiated with a concrete construction, e.g., the one given by~\cite{Guo2020a} is secure in the random permutation model with specific bounds depending on the number of queries made by the distinguisher. In garbling schemes, the number of queries roughly translates to the number of AND gates in the garbled circuit(s) that use the same offset. Since $n$-TCCR is a generalization of TCCR, we expect similar assumptions and limitations. We estimate an advantage of around $n$-bit for the distinguisher of $n$-TCCR compared to TCCR when using non-dedicated TCCR constructions. This means that, to attain a similar security level as TCCR, one may need to choose larger parameters for $n$-TCCR.

\subsection{Privacy}
\label{sec:security:prv.sim}
The prv.sim definition states that given the garbled circuit $GC$, all the labels of the garbled input $X$ and the decoding information $d$, no information is revealed about the input except from what can be deduced from the output $y$.

\begin{theorem}\label{thm:security}
	Given a $\bar{n}$-TCCR secure hash function $\hash$ and $\bar{n} \ll \kappa$, the garbling scheme $\scheme$ is prv.sim secure.
\end{theorem}

\begin{figure}
   	\centering
   	\begin{framed}
    \algrenewcommand\algorithmicindent{0.5em}%
   	\begin{algorithmic}[1]
   	    \small
   		\Function{$\Simulator$}{$f,y$}
   		\For{$i \in \Inputs$}
   			\State $W_i^{0^{\length(i)}} \sample \{0,1\}^k$
   			\State $X_i \gets W_i^{0^{\length(i)}}$
   		\EndFor
   		\For{$\G_i \in \Gates$}
   			\If{$\G_i = XOR$}
   				\Comment $\G_i$ with $n$-bit input wires $a$, $b$
   				\State $W_i^{0^n} \gets W_a^{0^n} \oplus W_b^{0^n}$
   			\Else
   				\Comment $\G_i$ with $n$-bit input wire $a$ and output\Statex\Comment size $m$ of the function $\G_i$ realizes
   				\State $W_i^{0^m} \sample \{0,1\}^k$
   				\State $GC[i,\lsb_n(W_a^{0^n})] \gets \hash_{n,m}(W_a^{0^n}, i) \oplus W_i^{0^m}$ \label{alg:security:simulator:active-path}
   			\For{$x \ne 0^n \in  \{0,1\}^n$}
   				\State $GC[i,\lsb_n(W_a^{0^n}) \oplus x]$
                \Statex \hspace*{\fill} $\gets \rand_{n,m}(W_a^{0^n},i,x,0^m) \oplus W_i^{0^m}$ \label{alg:security:simulator:proj-loop}
   			\EndFor
   			\EndIf
   		\EndFor
   		\For{$i \in \Outputs$} \label{alg:security:simulator-output-start}
   			\State $d_i \gets \lsb_n(W_i^{0^n}) \oplus y_i$
   		\EndFor \label{alg:security:simulator-output-end}
   		\State \Return $GC, X, d$
   		\EndFunction
   	\end{algorithmic}
   	\end{framed}
   	\caption{Simulator $\Simulator$.}
   	\label{fig:simulator}
\end{figure}
\begin{proof}
We define a simulator $\mathcal{S}$ (see Fig.~\ref{fig:simulator}) and show through a series of hybrids that the output of $\Simulator$ is indistinguishable for an adversary from the output of $\Garble$.
We require $\bar{n} \ll \kappa$, i.e., the largest bit length $\bar{n}$ used in a wire in the circuit is small compared to the security parameter $\kappa$, to ensure that for any adversarially chosen circuit, both the garbling scheme and the simulator run in polynomial time. In the following, we use $\rand_{n,m}(m,t,\vec{a},\vec{b})$ where $\vec{a} \in \{0,1\}^n$ and $\vec{b} \in \{0,1\}^m$ which can be constructed from $\bar{n}$-TCCR $\rand$ by padding the a/b inputs with zeros to reach the full length of $\bar{n}$, e.g., $\rand(m,t,\vec{a}||0^{\bar{n}-n}, \vec{b}||0^{\bar{n}-m})$ since $\bar{n} \ge n,m$ in the whole circuit. The same can be done for $\hash_{n,m}$ and $\hash$.
When evaluating a garbled circuit, let the assignment of active labels to the wires be called the active path, i.e., for input wires, the active labels are retrieved via OT, for gate outputs, the active wires are retrieved by decrypting the row denoted by the point-and-permute bits.
	
The idea of the simulator is to produce a garbled circuit with a fixed active path. The simulator chooses the wire labels such that
\begin{itemize}
	\item the garbled input $X$ that is handed to the adversary corresponds to $0^p$;
	\item the active label on each gate's output wire that the adversary obtains if they choose to evaluate the circuit with $X$ is $W^{0^n}$ (see Line~\ref{alg:security:simulator:active-path} in Fig.~\ref{fig:simulator}).
\end{itemize} 
The simulator adapts the decoding information s.t. if the garbled output is $W^{0^n}$, the expected output $y$ is decoded.
    
\begin{figure}
	\centering
	\begin{framed}
    \algrenewcommand\algorithmicindent{0.5em}%
	\begin{algorithmic}[1]
	    \small
		\Function{EvalWires}{$f,x$}
		\For{$i \in \Inputs$}
			\State $v_i \gets x_i$
		\EndFor
		\For{$\G_i \in \Gates$}
			\State $a,b \gets in(\G_i)$
			\If{$\G_i = XOR$}
				\Comment $\G_i$ with $n$-bit input wires $a$, $b$
				\State $v_i \gets v_a \oplus v_b$
			\Else
				\Comment $\G_i$ with $n$-bit input wire $a$
				\State $v_i \gets \phi(v_a)$ \Comment and $\phi: \{0,1\}^n \rightarrow \{0,1\}^m$
			\EndIf
		\EndFor
		\State \Return $v$
		\EndFunction
		\Statex
		
		\Function{$\Hybrid_1$}{$f,x$}
		\State $v \gets \textsf{EvalWires}(f,x)$
		\For{$i \in \Inputs$}
			\State $W_i^{\new{v_i}} \sample \{0,1\}^k$
			\State $X_i \gets W_i^{\new{v_i}}$
		\EndFor
		\For{$\G_i \in \Gates$}
			\If{$\G_i = XOR$}
				\Comment $\G_i$ with $n$-bit input wires $a$, $b$
				\State $W_i^{\new{v_i}} \gets W_a^{\new{v_a}} \oplus W_b^{\new{v_b}}$
			\Else
				\Comment $\G_i$ with $n$-bit input wire $a$
				\State $W_i^{\new{v_i}} \sample \{0,1\}^k$ \Comment and $\phi: \{0,1\}^n \rightarrow \{0,1\}^m$
				\State $GC[i,\lsb_n(W_a^{\new{v_a}})] \gets \hash_{n,m}(W_a^{\new{v_a}}, i) \oplus W_i^{\new{v_i}}$ \label{alg:security:hybrid3:exception}
			\For{$x \ne \new{v_a} \in \{0,1\}^n$}
				\State $GC[i,\lsb_n(\new{W_a^{v_a}}) \oplus x] \gets$
                \Statex \hfill $\rand_{n,m}(W_a^{\new{v_a}}, i, \new{v_a \oplus x}, \new{\phi(v_a \oplus x)}) \oplus W_i^{\new{v_i}}$
			\EndFor
			\EndIf
		\EndFor
		\For{$i \in \Outputs$}
			\State $d_i \gets \lsb_n(W_i^{\new{v_i}}) \oplus y_i$
		\EndFor
		\State \Return $GC, X, d$
		\EndFunction
	\end{algorithmic}
	\end{framed}
	\caption{Hybrid $\Hybrid_1$. The simulator from the perspective of the evaluator where $x$ is a black box value. Values in a box $\new{v_i}$ highlight the difference between $\Simulator$ and $\Hybrid_1$.}
	\label{fig:hybrid1}
\end{figure}

\begin{description}
	\item[$\mathcal{S} \approx \Hybrid_1$.]
	Hybrid $\Hybrid_1$ (see Fig.~\ref{fig:hybrid1}) describes the simulator from the perspective of the evaluator.
	Let $x$ be the input that the adversary chooses in the game. We view $x$ as a black box as it is unknown. Suppose we evaluated the circuit on $x$ in plaintext. We denote $v_i$ as the active value on wire $i$. Instead of fixing the active path on labels $W^{0^n}$, we fix it on $W^{v_i}$.\\
	The output values $GC,d$ and the outputs of $\Simulator$ are identically distributed as $W^{0^n}$ and $W^{v_i}$ are both distributed uniformly at random. Further, the change of input arguments,
     \begin{align*}
         &\rand_{n,m}(W_a^{0^n},i,x,0^m)\\
         &\qquad \approx \rand_{n,m}(W_a^{v_a}, i, v_a \oplus x, \phi(v_a \oplus x)) \enspace,
     \end{align*}
	does not change the distribution since all inputs $(x,0^m),\,\forall x \in \{0,1\}^n \ne 0^n$ and $(v_a \oplus x,\phi(v_a \oplus x)),\, \forall x \in \{0,1\}^n \ne v_a$, respectively, are unique and therefore amount to fresh randomness from the oracle, irrespective of $\phi$.
		
	\item[$\Hybrid_1 \approx \Hybrid_2$.]
	In hybrid $\Hybrid_2$, we replace $\rand_{n,m}$ by the real construction $\hash_{n,m}(m \oplus \vec{a}\cdot\vec{R}_n, t) \oplus \vec{b}\cdot\vec{R}_m$. This change is indistinguishable for the adversary by the definition of the $\bar{n}$-TCCR secure function $\hash$ (see Definition~\ref{def:n-tccr})
	\begin{align*} 
		&\rand_{n,m}(W_a^{v_a},i,v_a \oplus x,\phi(v_a \oplus x)) \\ 
        & \hfill \approx \hash_{n,m}(W_a^{v_a} \oplus (v_a \oplus x)\cdot \vec{R}_n, i) \oplus \phi(v_a \oplus x)\cdot\vec{R}_m\,.
    \end{align*}
	\item[$\Hybrid_2 \approx \Hybrid_3$.]
	In hybrid $\Hybrid_3$ (see Fig.~\ref{fig:hybrid3}), we no longer compute the wire values $v_i$ explicitly from the black-box input $x$. We fix an encoding for $v_i$, namely $v_i = 0^n$.\\
	For the input wires, note that $x_i = v_i$ by definition of \textsc{EvalWires}, so $X_i \gets W_i^{x_i}$ instead of $W_i^{v_i}$.\\
	Further, the ciphertext indexing $GC[i,\cdot]$ (Line~\ref{alg:security:hybrid3:loop} in Fig.~\ref{fig:hybrid3}) is identical after the re-write. In $\Hybrid_2$,
	\[
	    \lsb_n(W_a^{v_a}) \oplus x = \lsb_n(W_a^{0^n}) \oplus x\,,
	\]
	and in $\Hybrid_3$,
	\[
	    \lsb_n(W_a^x) = \lsb_n(W_a^{0^n}) \oplus \lsb_n(x \cdot \vec{R}_n) = \lsb_n(W_a^{0^n}) \oplus x
	\] by definition of $\vec{R}_n$.
	In the output of all gates $\G_i$, we now maintain the invariant with $x \in \{0,1\}^n$
	\begin{align*}
		W_i^{v_i} & \text{ becomes } W_i^{0^n}\,, \\
		W_i^{v_i} \oplus (v_i \oplus x) \cdot \vec{R}_{\length(i)} & \text{ becomes } W_i^{0^n} \oplus x\cdot\vec{R}_{\length(i)}\,.
	\end{align*}
	And for the decoding information, first note that for $i \in \Outputs~v_i = y_i$,
	thus in $\Hybrid_2$ (we denote $\length(i) = n$), 
	\[
	d_i \oplus \lsb_n(W_i^{vi}) = \lsb_n(W_i^{v_i}) \oplus \lsb_n(W_i^{v_i}) \oplus y_i = y_i \,,
	\]
	and in $\Hybrid_3$:
	\begin{align*}
		d_i \oplus \lsb_n(W_i^{y_i}) & = \lsb_n(W_i^{0^n}) \oplus \lsb_n(W_i^{y_i}) \\
		& = \lsb_n(W_i^{0^n}) \oplus \lsb_n(W_i^{0^n}) \oplus y_i \\
		& = y_i\,.
	\end{align*}
	The decoding information in $\Hybrid_2$ and $\Hybrid_3$ yield correct results when used with their respective garbled inputs. $d_{\Hybrid_2}$ and $d_{\Hybrid_3}$ are both uniformly distributed as $\lsb_{\length(i)}(W_i^{vi})$ resp. $\lsb_{\length(i)}(W_i^{0^{\length(i)}})$ are distributed at random. So $d_{\Hybrid_2}$ and $d_{\Hybrid_3}$ remain indistinguishable.
\end{description}
We conclude the proof by noting that $\Hybrid_3$ and \Garble~yield \emph{identical} outputs in the prv.sim game. This can easily be seen when the exceptional case for $x = 0^n$ (Line~\ref{alg:security:hybrid3:exception} in Fig.~\ref{fig:hybrid3}) in the projection gates part is incorporated into the loop and the computation of $d$ is re-written, $\Hybrid_3$ is a description of the \Garble~function.
\end{proof}

\begin{figure}
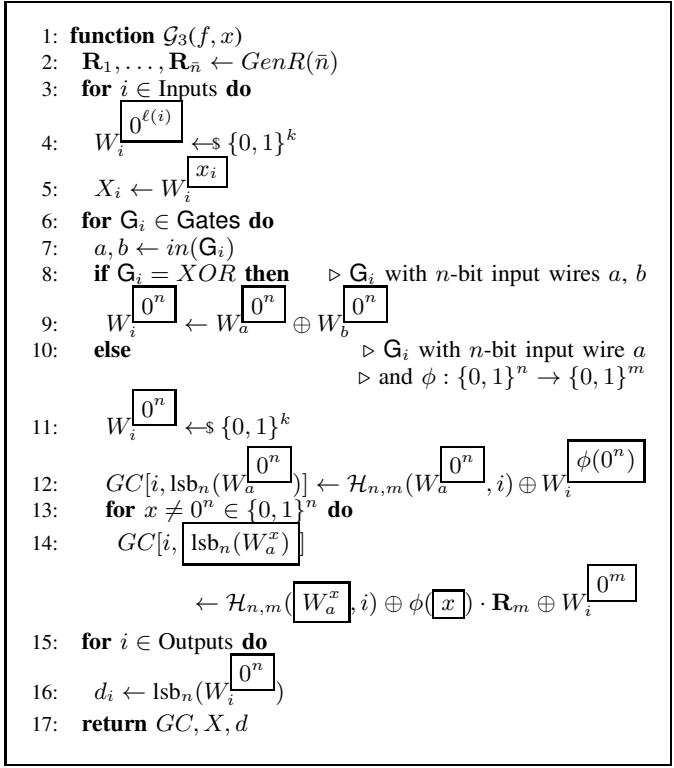

	\centering
	\begin{framed}\small
    \algrenewcommand\algorithmicindent{0.5em}%
	\begin{algorithmic}[1]
		\Function{$\Hybrid_3$}{$f,x$}
		\State $\vec{R}_1, \ldots, \vec{R}_{\bar{n}} \gets GenR(\bar{n})$
		\For{$i \in \Inputs$}
			\State $W_i^{\new{0^{\length(i)}}} \sample \{0,1\}^k$
			\State $X_i \gets W_i^{\new{x_i}}$
		\EndFor
		\For{$\G_i \in \Gates$}
			\State $a,b \gets in(\G_i)$
			\If{$\G_i = XOR$}
				\Comment $\G_i$ with $n$-bit input wires $a$, $b$
				\State $W_i^{\new{0^n}} \gets W_a^{\new{0^n}} \oplus W_b^{\new{0^n}}$
			\Else
				\Comment $\G_i$ with $n$-bit input wire $a$
                \Statex\Comment and $\phi: \{0,1\}^n \rightarrow \{0,1\}^m$
				\State $W_i^{\new{0^n}} \sample \{0,1\}^k$
				\State $GC[i,\lsb_n(W_a^{\new{0^n}})] \gets \hash_{n,m}(W_a^{\new{0^n}}, i) \oplus W_i^{\new{\phi(0^n)}}$
				\For{$x \ne 0^n \in \{0,1\}^n$}
					\State $GC[i,\new{\lsb_n(W_a^x)}]$
                    \Statex \hspace*{\fill} $\gets \hash_{n,m}(\new{W_a^x}, i) \oplus \phi(\new{x}) \cdot \vec{R}_m \oplus W_i^{\new{0^m}}$ \label{alg:security:hybrid3:loop}
				\EndFor
			\EndIf
		\EndFor
		\For{$i \in \Outputs$}
			\State $d_i \gets \lsb_n(W_i^{\new{0^n}})$ \del{$\oplus~y_i$}
		\EndFor
		\State \Return $GC, X, d$
		\EndFunction
	\end{algorithmic}
	\end{framed}%
	\caption{Hybrid $\Hybrid_3$. We fix the encoding of $W_i^{v_i}$ to $W_i^{0^n}$. Values in a box $\new{0^n}$ highlight the difference between $\Hybrid_2$ and $\Hybrid_3$.}
	\label{fig:hybrid3}
\end{figure}

\subsection{Obliviousness}
\label{sec:security:obv.sim}
The notion of obv.sim expresses that the adversary cannot learn any information given the garbled circuit $GC$ and all input wire labels $X$. Unlike the privacy notion, the adversary does not have access to the decoding information $d$.
\begin{theorem}
	Given a $\bar{n}$-TCCR secure hash function $\hash$ and $\bar{n} \ll \kappa$, the garbling scheme $\scheme$ is obv.sim secure.
\end{theorem}
\begin{proof}
	Let $\Simulator_{\textup{auth}}$ be $\Simulator$ from Fig.~\ref{fig:simulator} with the lines~\ref{alg:security:simulator-output-start}-\ref{alg:security:simulator-output-end} removed. Then we note that the computation of $GC$ and $X$ doesn't depend on $y$, neither in $\Simulator_{\textup{auth}}$ nor in one of the hybrids $\Hybrid_1, \Hybrid_2, \Hybrid_3$. We can thus use the same reasoning as for prv.sim security, omitting parts that correspond to $y$ or $d. \hfill$
\end{proof}

\subsection{Authenticity}
\label{sec:security:authenticity}
Authenticity states that an adversary cannot forge wire labels that are not obtained through evaluating the garbled circuit. Clearly, the presented scheme does not satisfy this property as any wire label is decoded to output bits. If authenticity is desired, we modify the decoding information $d$ to list hashes of all output wire labels and associations to their semantic meaning. As in \cite{Bellare2012}, the decoding function checks if the presented wire is indeed in the list $d$.
\section{Evaluation of SPN Primitives}
\label{sec:applications}
In the following, we discuss how SPN primitives with a specific structure can be implemented with our new garbling scheme and how this improves over the state-of-the-art. 
Note that we don't intend to compare the performance of the primitives among each other in MPC protocols. Instead, we focus on how each primitive can be accelerated. Consequently, we will not consider other traditional or MPC-friendly primitives, e.g.,~\cite{Albrecht2015}.
We compare the state-of-the-art garbling schemes Half-Gates~\cite{Zahur2015} as well as ThreeHalves~\cite{Rosulek2021}. Both schemes support free XOR gates and AND gates on wires holding one bit.

In SPN-based primitives, a state is updated with a round function consisting of a substitution layer, a permutation layer, a round constant and/or (round) key addition layer. SPNs are commonly used to construct block ciphers and pseudo-random permutations used, e.g., in hash functions or MAC algorithms.

We show an efficient circuit representation with projection gates for primitives that satisfy the following conditions for state and round function parts.
\begin{itemize}
	\item \textbf{State.} The state is (conceptually) split into $n$-bit cells.
	\item \textbf{Substitution Layer.} The substitution layer consists of S-boxes that are applied to each cell.
	\item \textbf{Permutation Layer.} The permutation layer can be described by a permutation on the cells and/or by a mixing matrix which encodes a fixed matrix multiplication with the state. In this paper, we focus on primitives with a binary  mixing matrix. 
	\item \textbf{Round Constant/(Round) Key Addition Layer.} The round constant or (round) key is XORed cell-wise.
\end{itemize}
With this structure, we set $\bar{n} = n$ and implement a single cell as $n$-bit wire. Each S-box in the substitution layer is replaced with an $n$-bit projection gate computing the same functionality. The permutation layer and the addition layer are expressible with XOR gates only. We illustrate how the SKINNY 4-bit S-box is expressed in our scheme and as a Boolean circuit in Fig.~\ref{fig:applications:example} as an example.

We identified nine SPN primitives in the literature that fulfill the conditions. Since the studied primitives have at most 8-bit cells, we set $\bar{n} = 8$.
These primitives include the widely-used AES standard and several established second-round and finalists of the NIST Lightweight Cryptography Competition, each offering unique hardware, performance, and energy efficiency metrics for specific real-world use cases.

\begin{figure}
    \begin{subfigure}{\linewidth}
        \centering
        \begin{tikzpicture}
        \node[rectangle, draw] (proj) {$\Proj_\phi$};
        \coordinate[left=1.5cm of proj.west] (input);
        \coordinate[right=2cm of proj.east] (output);
        \draw (input) -- node[above,xshift=-1cm]{\scriptsize$W_a^0 \oplus b_1 R_1 \oplus \dots \oplus b_3 R_3$} (proj);
        \draw (proj) --  node[above,xshift=0.7cm]{\scriptsize$W_b^0 \oplus b_1' R_1 \oplus \dots \oplus b_3' R_3$}(output);
    \end{tikzpicture}
    \caption{Implementation of the SKINNY 4-bit S-box via a 4-to-4-bit projection gate in our garbling scheme. Note that $\phi$ is the S-box lookup function.}
    \label{fig:applications:example:proj}
    \end{subfigure}

    \begin{subfigure}{\linewidth}
        \centering
        \begin{tikzpicture}[circuit logic US]
        \tikzstyle{dot}=[draw,circle,fill=black,inner sep=1pt]
        \matrix[column sep=2ex, row sep=2ex]
        {
            \coordinate (b0); & \node[xor gate,rotate=90] (xor1) {}; & \coordinate (xor1 output); & & \node[xor gate,rotate=90] (xor2) {}; & \coordinate (xor2 output); & & \node[xor gate,rotate=90] (xor3) {}; & \coordinate (xor3 output); & & \node[xor gate,rotate=90] (xor4) {}; & \coordinate (t0); \\
            \coordinate (b1); & & \coordinate (b1 output); & \coordinate (xor1 output2); & \coordinate (xor1 output3); & \coordinate (xor2 output2); & & & \coordinate (xor3 output2); & & & \coordinate (t1); \\
            & \node[and gate,inputs={ii},rotate=90] (and1) {}; & \coordinate (b2 output); & & \node[and gate,inputs={ii},rotate=90] (and2) {}; & \coordinate (and2 output); & & \node[and gate,inputs={ii},rotate=90] (and3) {}; & \coordinate (and3 output); & & \node[and gate,inputs={ii},rotate=90] (and4) {}; &\\
            \coordinate (b2); & & \coordinate (b2 output); & \coordinate (and2 input2); & & \coordinate (col2 output1); & \coordinate (and3 input2); & & \coordinate (col3 output1); & \coordinate (and4 input2); & & \coordinate (t2); \\
            \coordinate (b3); & & \coordinate (b3 output); & \coordinate (and2 input1); & & \coordinate (col2 output2); & \coordinate (and3 input1); & & \coordinate (col3 output2); & \coordinate (and4 input1); & & \coordinate (t3); \\
        };
        
        \draw ([xshift=-2ex] b0) node[above] {\scriptsize$W_a^0 \oplus b_1R$} -- (b0) -- ++(right:1mm) -| ++(south:4ex) -| (xor1.input 1);
        \draw (and1.output) -- ++(north:1mm) -| (xor1.input 2);
        \draw ([xshift=-2ex] b2) node[above] {\scriptsize $W_c^0 \oplus b_3R$} -- (b2) -| node[pos=0.5,dot] {} (and1.input 2);
        \draw ([xshift=-2ex] b3) node[below] {\scriptsize $W_d^0 \oplus b_4R$} -- (b3) -| node[pos=0.5,dot] {} (and1.input 1);

        \draw (xor1.output) -- ++(north:1mm) -- ++(right:1ex) -- (xor1 output2);
        \draw ([xshift=-2ex] b1) node[below] {\scriptsize $W_b^0 \oplus b_2R$} -- (b1) -- (b1 output) -- (and2 input2) -| node[pos=0.5,dot] {} (and2.input 2);
        \draw (b2) -- (b2 output) -- (and2 input1) -| node[pos=0.5,dot]{} (and2.input 1); 
        \draw (b3) -- (b3 output) -- ([yshift=-1mm] xor2.input 1) -- (xor2.input 1);
        \draw (and2.output) -- ++(north:1mm) -| (xor2.input 2);

        \draw (xor2.output) -- ++(north:1mm) -- ++(right:1ex) -- (xor2 output2);
        \draw (xor1 output2) -- ([xshift=-1ex] xor2 output2) -- (and3 input2) -| node[pos=0.5,dot] {} (and3.input 2);
        \draw (and2.input 2) |- (col2 output1) -- (and3 input1) -| node[pos=0.5,dot] {} (and3.input 1);
        \draw (and2 input1) -- (col2 output2) -- ([yshift=-1mm] xor3.input 1) -- (xor3.input 1);
        \draw (and3.output) -- ++(north:1mm) -| (xor3.input 2);

        \draw (xor3.output) -- ++(north:1mm) -- ++(right:1ex) -- (xor3 output2);
        \draw (xor2 output2) -- ([xshift=-1ex] xor3 output2) -- (and4 input2) -| node[pos=0.5,dot] {} (and4.input 2);
        \draw (and3.input 2) |- (col3 output1) -- (and4 input1) -| node[pos=0.5,dot] {} (and4.input 1);
        \draw (and3 input1) -- (col3 output2) -- ([yshift=-1mm] xor4.input 1) -- (xor4.input 1);
        \draw (and4.output) -- ++(north:1mm) -| (xor4.input 2);

        \draw (xor4.output) -- ++(north:1mm) -- ++(right:1ex) -- (t0) -- node[below,xshift=2mm] {\scriptsize $W_e^0 \oplus b_1'R$} ++(right:2ex);
        \draw (xor3 output2) -- (t1) -- node[below,xshift=2mm] {\scriptsize $W_f^0 \oplus b_2'R$} ++(right:2ex);
        \draw (and4 input2) -- (t2) -- node[above,xshift=2mm] {\scriptsize $W_g^0 \oplus b_3'R$} ++(right:2ex);
        \draw (and4 input1) -- (t3) -- node[below,xshift=2mm] {\scriptsize $W_h^0 \oplus b_4'R$} ++(right:2ex);
    \end{tikzpicture}
    \caption{Implementation of the SKINNY 4-bit S-box via Boolean circuit with input wires $a,b,c,d$ and output wires $e,f,g,h$. $b_1, \dots, b_4$ denote the 4 input bits from least to most significant bit. $b_1', \dots, b_4'$ denote the 4 output bits in the same order. Circuit from~\cite{Beierle2016}.}
    \label{fig:applications:example:boolean}
    \end{subfigure}
    \caption{Example to illustrate the implementation difference between our scheme and a Boolean circuit for the 4-bit SKINNY S-box. Note that Fig.~\ref{fig:applications:example:proj} encodes 4 bits \emph{per} wire while Fig.~\ref{fig:applications:example:boolean} uses 4 wires.}
    \label{fig:applications:example}
\end{figure}
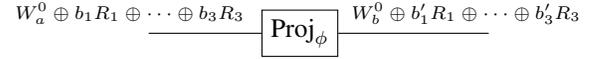
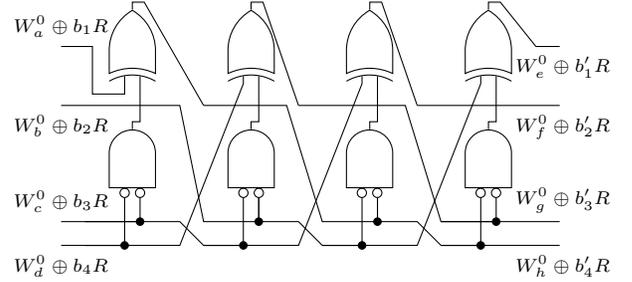

\subsection{Implementation Details}
\label{sec:applications:spn:data-path}
For the projection gates implementation, we assume that the input block is already setup in $n$-bit wires where $n$ denotes the cell size in bits. This doesn't incur additional cost since the input phase using OT can already share wire labels with the desired wire label offset, as detailed in Sect.~\ref{sec:scheme:ot}.
For the implementation using AND gates, only the S-box costs AND gates in the data path of the primitive. We selected implementations for the S-boxes with the lowest number of AND gates since their cost dominates in Half-Gates and ThreeHalves.
Table~\ref{tab:applications:spn:cipher-gate-counts} details the number of projection and AND gates for each primitive.

For the 4-bit S-box used in TWINE-80 and TWINE-128 and for the 4-bit S-box used in Midori64, MANTIS and CRAFT, we found new circuits using the smallest number of AND gates so far, reducing the number of AND gates for the TWINE S-box from 7 to 6 and for the Midori64 S-box from 8 to 4. 
We used  the heuristic optimization tool LIGHTER by Jean et al.~\cite{Jean2017} operating on a customized cost metric, for more details see Appendix~\ref{sec:appendix:sbox-formulas}.

For the key, we assume individual key bits to be available in $1$-bit wires as this eases key scheduling in many cases. Note that the cost to transform the (round) key bits into $n$-bit wires is taken into account.
In scenarios where one party knows the complete key, e.g., to offer blind symmetric encryption or decryption where the encryption or decryption is performed without learning the message and ciphertext, the key schedule does not need to be computed within the garbling scheme. 
Instead, if the garbler knows the key, they can compute the key schedule separately and insert the round keys as secret constants.
Similarly, if the evaluator knows the key, they may receive the wire labels for round keys via OT instead.

If the key is shared among the players using a linear secret-sharing scheme, for instance as $k = k_G \oplus k_E$ where $k_G$ is the garbler's share and $k_E$ is the evaluator's share, the key schedule can be computed outside of the garbling scheme by each player on their share instead for ciphers with a linear key schedule, e.g., for Piccolo, Midori, SKINNY, MANTIS and CRAFT. The resulting round key shares can then be treated as input and are recombined using only linear operations saving any gates specified in the key schedule column for the cipher. However, the gate counts presented here compute the entire key schedule of the primitive which is required in the distributed encryption/decryption scenario.

\begin{table}\scriptsize
    \setlength{\tabcolsep}{1em}
	\centering
	\caption{Detailed gate counts for setup, key schedule and data path of the selected symmetric primitives. The top entry denotes the number of AND gates while the bottom entry denotes the number of projection gates.\\
	{\footnotesize $\dagger$ Gate counts obtained from Mandal et al.~\cite{Mandal2020}.}
	}
	\label{tab:applications:spn:cipher-gate-counts}
	\begin{tabular}{llll}
		\toprule
		Primitive & Setup & Key Schedule & Data Path \\
		\midrule
        \multirow{2}{*}{AES-128~\cite{NIST2001}} & & 1280 AND & 5120 AND \\
		  & & 128 1-bit + 49 8-bit & 320 8-bit \\
		\midrule
        \multirow{2}{*}{CRAFT~\cite{Beierle2019}} & & & 1920 AND \\
		  & & 192 1-bit & 480 4-bit \\
        \midrule
        \multirow{2}{*}{Fides-80~\cite{Bilgin2013}} & & & 320 AND\\
		  & 160 1-bit & & 32 5-bit \\
		\midrule
		\multirow{2}{*}{Fides-96} & &  & 1088 AND\\
		  & 192 1-bit & & 32 6-bit \\
		\midrule
        \multirow{2}{*}{MANTIS~\cite{Beierle2016}} & &  & 896 AND \\
		  &  & 192 1-bit & 224 4-bit \\
        \midrule
        \multirow{2}{*}{Midori64~\cite{Banik2015}} & &  & 1024 AND \\
		  & & 128 1-bit & 256 4-bit \\
		\midrule
        \multirow{2}{*}{Piccolo-80~\cite{Shibutani2011}} & & & 1600 AND \\
		  & & 80 1-bit & 600 4-bit \\
		\midrule
		\multirow{2}{*}{Piccolo-128} & & & 1984 AND \\
		  & & 128 1-bit & 744 4-bit \\
		\midrule
		\multirow{2}{*}{SKINNY-64-128} & & & 2304 AND \\
		  & & 128 1-bit + 280 4-bit & 576 4-bit \\
	    \midrule
		\multirow{2}{*}{TWINE-80~\cite{Suzaki2012}} & & 432 AND & 1728 AND \\
		& & 80 1-bit + 70 4-bit & 288 4-bit \\
		\midrule
		\multirow{2}{*}{TWINE-128} & & 630 AND & 1728 AND \\
		  & & 128 1-bit + 104 4-bit & 288 4-bit \\
		\midrule
		\multirow{2}{*}{WAGE~\cite{AlTawy2020}} & & & 37745 AND$^\dagger$ \\
		  & 259 1-bit & & 777 7-bit \\
		\bottomrule
	\end{tabular}
\end{table}

\begin{table}\scriptsize
    \setlength{\tabcolsep}{1em}
	\centering
	\caption{Estimated performance difference for selected symmetric ciphers. The notation $\times x$ denotes an improvement by factor $x$ in the category with respect to the base scheme, i.e., $x > 1$ is an improvement, $x < 1$ is degradation.}
	\label{tab:applications:spn:ciphers}
	\begin{tabular}{lllll}
		\toprule
		Base Scheme & Primitive & Garble & Send & Eval \\
		\midrule
        Half-Gates~\cite{Zahur2015} & \multirow{2}{*}{AES-128~\cite{NIST2001}} & $\times0.28$ & $\times0.14$ & $\times26.23$ \\
		ThreeHalves~\cite{Rosulek2021} & & $\times0.42$ & $\times0.10$ & $\times39.34$ \\
		\midrule
        Half-Gates & \multirow{2}{*}{CRAFT~\cite{Beierle2019}} & $\times0.95$ & $\times0.52$ & $\times5.71$ \\
		ThreeHalves & & $\times1.43$ & $\times0.39$ & $\times8.57$ \\
        \midrule
        Half-Gates & \multirow{2}{*}{Fides-80~\cite{Bilgin2013}} & $\times1.23$ & $\times0.64$ & $\times15.45$ \\
		ThreeHalves & & $\times1.84$ & $\times0.48$ & $\times23.18$ \\
		\midrule
		Half-Gates & \multirow{2}{*}{Fides-96} & $\times2.10$ & $\times1.07$ & $\times50.26$ \\
		ThreeHalves & & $\times3.15$ & $\times0.81$ & $\times75.39$ \\
		\midrule
        Half-Gates & \multirow{2}{*}{MANTIS~\cite{Beierle2016}} & $\times0.90$ & $\times0.50$ & $\times4.31$ \\
		ThreeHalves & & $\times1.35$ & $\times0.38$ & $\times6.46$ \\
		\midrule
        Half-Gates & \multirow{2}{*}{Midori64~\cite{Banik2015}} & $\times0.94$ & $\times0.52$ & $\times5.33$ \\
		ThreeHalves & & $\times1.41$ & $\times0.39$ & $\times8.00$ \\
		\midrule
        Half-Gates & \multirow{2}{*}{Piccolo-80~\cite{Shibutani2011}} & $\times0.66$ & $\times0.35$ & $\times4.71$ \\
		ThreeHalves & & $\times0.98$ & $\times0.26$ & $\times7.06$ \\
		\midrule
		Half-Gates & \multirow{2}{*}{Piccolo-128} & $\times0.65$ & $\times0.35$ & $\times4.55$ \\
		ThreeHalves & & $\times0.98$ & $\times0.26$ & $\times6.83$ \\
		\midrule
		Half-Gates & \multirow{2}{*}{SKINNY-64-128} & $\times0.66$ & $\times0.36$ & $\times4.68$ \\
		ThreeHalves & & $\times0.99$ & $\times0.27$ & $\times7.02$ \\
		\midrule
		Half-Gates & \multirow{2}{*}{TWINE-80~\cite{Suzaki2012}} & $\times1.46$ & $\times0.79$ & $\times9.81$ \\
		ThreeHalves & & $\times2.19$ & $\times0.59$ & $\times14.71$ \\
		\midrule
		Half-Gates & \multirow{2}{*}{TWINE-128} & $\times1.44$ & $\times0.78$ & $\times9.05$ \\
		ThreeHalves & & $\times2.16$ & $\times0.59$ & $\times13.58$ \\
		\midrule
		Half-Gates & \multirow{2}{*}{WAGE~\cite{AlTawy2020}} & $\times1.51$ & $\times0.76$ & $\times72.87$ \\
		ThreeHalves & & $\times2.27$ & $\times0.57$ & $\times109.30$ \\
		\bottomrule
	\end{tabular}
\end{table}

\subsection{Performance}\label{sec:applications:performance}
\setlength{\tabcolsep}{.75em}
\begin{table}\scriptsize
    \centering
    \caption{Performance benchmark results for some SPN-ciphers comparing garbling and evaluation time as well as the circuit size. All reported numbers are amortized from 500 (for SKINNY-128-*) and 1000 parallel primitive calls averaged over 10 repetitions.}
    \label{tab:applications:spn:implementations}
    \begin{tabular}{lrrrr}
        \toprule
        Base Scheme & Primitive & Garble & Circuit size & Eval \\
        & & in ms & in KB & in ms \\
        \midrule
        Half-Gates~\cite{Zahur2015} & \multirow{3}{*}{AES-128~\cite{NIST2001}} & 0.767 & 204.93 & 0.722 \\
        ThreeHalves~\cite{Rosulek2021} & & \textbf{0.436} & \textbf{156.51} & 0.305 \\
        This work & & 0.928 & 1242.55 & \textbf{0.016} \\
        \midrule
        Half-Gates & \multirow{3}{*}{MANTIS$_7$~\cite{Beierle2016}} & 0.093 & 32.84 & 0.070 \\
        ThreeHalves & & \textbf{0.083} & \textbf{30.96} & 0.077 \\
        This work & & 0.133 & 76.36 & \textbf{0.040} \\
        \midrule
        Half-Gates & \multirow{3}{*}{SKINNY-64-128} & 0.283 & 73.80 & 0.194 \\
        ThreeHalves & & \textbf{0.153} & \textbf{59.13} & 0.096 \\
        This work & & 0.289 & 139.30 & \textbf{0.026} \\
        \midrule
        Half-Gates & \multirow{3}{*}{SKINNY-64-192} & 0.343 & 81.99 & 0.246 \\
        ThreeHalves & & \textbf{0.166} & \textbf{67.11} & 0.118 \\
        This work & & 0.321 & 154.69 & \textbf{0.041} \\
        \midrule
        Half-Gates & \multirow{3}{*}{SKINNY-128-128} & 0.595 & 163.98 & 0.440 \\
        ThreeHalves & & \textbf{0.346} & \textbf{126.66} & 0.279 \\
        This work & & 2.281 & 2613.28 & \textbf{0.015} \\
        \midrule
        Half-Gates & \multirow{3}{*}{SKINNY-128-256} & 0.803 & 196.74 & 0.594 \\
        ThreeHalves & & \textbf{0.442} & \textbf{154.71} & 0.348 \\
        This work & & 2.563 & 3135.62 & \textbf{0.028} \\
        \midrule
        Half-Gates & \multirow{3}{*}{SKINNY-128-384} & 1.107 & 229.51 & 0.819 \\
        ThreeHalves & & \textbf{0.558} & \textbf{182.77} & 0.472 \\
        This work & & 2.841 & 3658.00 & \textbf{0.041} \\
        \midrule
        Half-Gates & \multirow{3}{*}{TWINE-128~\cite{Suzaki2012}} & 0.202 & 75.52 & 0.168 \\
        ThreeHalves & & \textbf{0.136} & \textbf{60.40} & 0.081 \\
        This work & & 0.191 & 108.21 & \textbf{0.059} \\
        \midrule
        Half-Gates & \multirow{3}{*}{TWINE-80} & 0.187 & 68.80 & 0.153 \\
        ThreeHalves & & \textbf{0.128} & \textbf{53.99} & 0.074 \\
        This work & & 0.199 & 99.04 & \textbf{0.045} \\
        \bottomrule
    \end{tabular}
\end{table}

\subsubsection{Instantiating $\hash$}
Before comparing our garbling scheme with Half-Gates and ThreeHalves, we need to discuss the instantiation of the $n$-TCCR hash function $\hash$. Half-Gates uses a construction called TCCR for naturally derived keys, while ThreeHalves uses a randomized TCCR function. Despite minor differences, both constructions use one call to an ideal permutation (which is instantiated using fixed-key AES), which is the main computation cost. In the following comparison, we also assume that $\hash$ can be instantiated in a construction using only one ideal permutation call. For the practical benchmark, we use the same construction as Half-Gates.
However, we need to stress a major difference: Since our wire labels are $\kappa+n$ bits long and we still want to fit them into one permutation call, the security of our scheme is reduced by $n$ bits with this instantiation. Concretely, since all studied primitives have $\bar{n} \le 8$, Half-Gates and ThreeHalves have $\kappa = 127$ while our scheme has $\kappa=120$ for AES with 128-bit block size.
Should a larger security level be desired, it is also possible to encrypt the pointer bits separately with another permutation call using a dedicated construction of $\hash$ that has an output length of $\kappa+n$. We leave such construction and the required dedicated security analysis to future work.

\subsubsection{Evaluation Performance}
The gate counts from Table~\ref{tab:applications:spn:cipher-gate-counts} can be turned into calls to $\hash$ and sent ciphertexts. In Half-Gates, each AND gate costs 4 calls to $\hash$ for garbling, 2 ciphertexts are sent, and 2 calls to $\hash$ for evaluation. In ThreeHalves, each AND gate costs 6 calls to $\hash$ for garbling, 1.5 ciphertexts are sent, and 3 calls to $\hash$ for evaluation.

Table~\ref{tab:applications:spn:ciphers} lists all studied primitives with the corresponding trade-off in garbling and communication cost, and evaluation improvement measured in the number of calls to $\hash$ and in the number of ciphertexts, respectively. We found three primitives in five configurations in total where our scheme improves in both garbling and evaluation cost over both reference garbling schemes\footnote{While the evaluation improvement is the goal of this work, the improvement in garbling time is caused by large, sub-optimal circuit representations of the 5-, 6- and 7-bit S-boxes for which, to the best of our knowledge, no smaller circuits are published in the literature.}. 
In the remaining primitives and cases, projection gates trade off higher garbling and communication cost for faster evaluation performance. Note that for most primitives, the evaluation improvement is much higher than the additional communication cost. E.g., for Midori64, at a cost of slightly more garbling work ($\approx 6\%$ more) and less than twice the number of sent ciphertexts, we improve the evaluation work by a factor of five. We detail the implementation approach with projection gates for the ciphers in Appendix~\ref{sec:appendix:ciphers}.

\subsubsection{Experimental Results}
Next, we experimentally compared the performance of four primitives in nine configurations in Half-Gates, ThreeHalves and our scheme. ThreeHalves has been implemented by Hamacher et al.~\cite{Hamacher2022} in the MOTION framework while Half-Gates and our scheme have been implemented in MP-SPDZ~\cite{Keller2020}. We also use the TMMO hash function construction that is already used for Half-Gates in MP-SPDZ for our scheme. Our code is publicly available\footnote{\codelocation}.
All used implementations already perform multi-threading\footnote{Garbler and evaluator are restricted to 4 threads each.} to accelerate the computation of gate garbling/evaluation as well as grouping together AES calls for improved pipe-lining behaviour. Batched/vectorized AES~\cite{Muench2021} may provide additional speed-up. The AES calls for our scheme are independent within each gate while the regular structure of the SPN primitives, e.g., all S-boxes are parallel, also allows simple batching for Half-Gates and ThreeHalves. We do not expect a drastic difference in how our scheme compares to Half-Gates and ThreeHalves when using vectorized AES.
Table~\ref{tab:applications:spn:implementations} lists the garbling and evaluation time, and the circuit size. Garbling and evaluation time are wall-clock running times as reported by the frameworks, and circuit size is the number of bytes sent/received (whichever is higher) as reported by the frameworks. Garbler and evaluator were run on the same machine (6-core/12-threads AMD Ryzen 5 PRO 4650U 2.1 GHz with 8 GB RAM), connected over localhost, and limited to 4 threads each. We achieve a considerable speed-up in evaluation time of, e.g., factor 20 to 45 for AES. The expected trade-off of faster evaluation and larger circuit size is immediate for all implemented ciphers. Even though they were benchmarked in the same test environment and hardware, we observed differences in garbling and evaluation time between Half-Gates and ThreeHalves executions of the same circuit which cannot be explained by the differing number of hash function calls. We believe the observations are due to the implementation in the two MPC frameworks which have differing overhead.

\subsection{Applications}\label{sec:applications:details}
In this section, we briefly describe two potential applications where fast evaluation of SPN primitives is crucial. 
\subsubsection{IoT-to-Cloud Secure Computation}
In an Internet of Things (IoT) to Cloud scenario, the focus is on encrypting data at the source, specifically on the IoT devices, employing an SPN primitive and efficient distributed decryption in the cloud prior to privacy-preserving computation on the data (cf.~\cite[Application 4]{Pinkas2009}). Such setup facilitates end-to-end secure data collection and processing. 
The distributed decryption should have low computation latency (and thus fast evaluation of the garbled circuit) in order to minimize the online phase of the combined protocol (distributed decryption and processing).
With our proposed garbling scheme, IoT-friendly lightweight primitives can be used for efficient encryption on resource-constraint devices and still obtain fast evaluation times for distributed decryption if garbled circuits are pre-processed.

\subsubsection{Distributed Kerberos} The Kerberos authentication protocol is widely used, e.g., as one authentication method for Microsoft Windows. The key distribution center (KDC) plays a critical role in the architecture since it stores the secret keys of clients and service servers, and the secret key of the ticket-granting server (TGS). During the different stages of authentication, the KDC decrypts messages using the TGS secret key, and also encrypts messages with the client, TGS, and service server secret keys.
Compromise of the KDC has devastating consequences as critical key material is leaked, which breaks all authentication goals. The security of keys can be increased by distributing the KDC among two servers $G,E$ that each only stores a secret-share of each stored key, e.g., $k = k_G \oplus k_E$. Consequently, the necessary encryption and decryption operations are computed using our garbling scheme where the two servers input their respective key share as private input. Clearly, clients only authenticate to the distributed KDC infrequently, allowing the garbling server $G$ to garble circuits for all required operations
\begin{itemize}
    \item encryption of the client session key using a client key (as input),
    \item encryption of the ticket granting ticket using the TGS secret key ($G$ can fix $k_G$ in the circuit),
    \item encryption of the client-service session key using a service server's key (as input) and
    \item decryption of the client session-key using the TGS secret key ($G$ can again fix $k_G$ in the circuit)
\end{itemize}
ahead of time and send the garbled circuit to $E$. The necessary setup for OT extensions can be performed similarly.
The resulting online phase either requires no input from $G$ or just the respective key share of the client/service server in question, and $E$ obtains the output without interaction with $G$. Therefore, a client only needs to send requests to $E$ and benefits from reduced circuit evaluation time to complete the entity authentication. The obtained speed-up of primitive evaluation for, e.g., AES, directly translates to its use in modes of operations specified in Kerberos.
The KDC has already been distributed in the three-party honest-majority setting~\cite{Araki2016} with the goal of high throughput. The authors used AES in counter mode to improve performance. Since garbled circuits are constant-round protocols, we can use Kerberos' original specification, AES-CTS (CBC ciphertext stealing).
\section{Conclusion}
\label{sec:conclusion}
We presented a garbling scheme that encodes $n$-bit strings per wire. It generalizes the idea of FreeXOR and integrates seamlessly into state-of-the-art schemes with FreeXOR on the 1-bit wire level. Projection gates can be used to convert strings from $n$- to $m$-bit or to compute arbitrary $n$- to $m$-bit functions, while XOR is free.
We prove the scheme secure under the assumption of a $n$-TCCR secure hash function $\hash$, a generalization of TCCR security. Instantiating $\hash$ with a dedicated construction is an interesting open direction since it requires input and output of $\kappa+n$ bits. One possible approach is concatenating two 128-bit constructions to produce sufficient output bits but care must be taken to keep the inputs distinct such as using different tweaks and/or chaining. It is also possible to instantiate current TCCR constructions with a permutation of larger size, e.g., fixed-key Rijndael-256, and still benefit from AES-NI hardware support.

For an important application in two-party secure function evaluation, the evaluation of symmetric primitives, we show that substitution-permutation network primitives with a cell-like structure can be efficiently implemented in our scheme. Compared to AND gate-based circuits, we show a high-speed evaluation that is traded off with moderate additional garbling or communication cost.  In scenarios where the garbling scheme runs in an offline/online setting, we shift the garbling work and garbled circuit transfer to the evaluator into the pre-processing phase and thus obtain a high-speed online phase.
We obtained a considerable performance improvement, a 4- to 72-times faster online phase, for nine primitives in literature when taking hash function calls as a metric. Implementation of some ciphers shows that this evaluation performance improvement translates into practical applications.

Besides oblivious computation of SPN primitives, statements where a prover proves knowledge of a key $k$ to a pair $x,y$ s.t. $AES_k(x) = y$ are highly relevant. Garbling schemes have been used to construct efficient interactive zero-knowledge protocols that prove statements over ``unstructured'' languages expressible in Boolean circuits~\cite{Jawurek2013,Frederiksen2015}. Using our garbling scheme, proving statements involving SPN primitives would be much faster since proving equates to evaluating the garbled circuit. This is traded-off with a larger proof size.

{\appendices \section{Formulas for S-Boxes of TWINE and Midori64}\label{sec:appendix:sbox-formulas}
We use  the heuristic optimization tool LIGHTER by Jean et al.~\cite{Jean2017} operating on a customized cost metric. We restrict the tool to use only NOT, AND and XOR gates with the associated costs of 0.01, 1 and 0.01, respectively. These costs describe our setting where NOT and XOR gates are practically free, i.e., very low cost, and AND gates are expensive, i.e., high cost\footnote{Essentially, this implies that we prefer implementations using 99 NOT or XOR gates in addition to $x$ AND gates to an implementation using $x+1$ AND gates.}. The tool then searches an implementation with low total cost following a heuristic. 
This approach reduces the number of AND gates for the TWINE S-box from 7 AND gates (algebraic normal form) to 6 AND gates (see Fig.~\ref{fig:applications:spn:twine-sbox}).
For the Midori64 S-box, the number of AND gates is reduced from 8 AND gates (formula given in the specification~\cite{Banik2015}) to 4 AND gates (see Fig.~\ref{fig:applications:spn:midori-sbox}).
\begin{figure}
    \begin{subfigure}[t]{0.4\columnwidth}\small
        \begin{align*}
            a & \gets x_2 \oplus x_3 \\
    		b & \gets x_3 \oplus (\neg x_0 \land x_1) \\
    		c & \gets \neg x_0 \oplus a \oplus (x_1 \land a \land b) \\
    		d & \gets a \oplus (b \land c) \\
    		e & \gets b \oplus c \\
            f & \gets c \oplus d \\
    		x_3' & \gets x_1 \oplus e \\
    		x_2' & \gets e \oplus (d \land x_3') \\
    		x_0' & \gets c \oplus (f \land x_2') \\
    		x_1' & \gets f \oplus x_3'
        \end{align*}
        \caption{The 4-bit S-box of TWINE~\cite{Suzaki2012} computed using 6 AND gates.}
        \label{fig:applications:spn:twine-sbox}
    \end{subfigure}
    \hspace{0.01\textwidth}
    \begin{subfigure}[t]{0.4\columnwidth}\small
        \begin{align*}
            & \\
    		a & \gets \neg (x_0 \oplus x_2) \\
    		b & \gets x_0 \oplus (a \land x_3) \\
    		c & \gets x_1 \oplus b \\
    		d & \gets \neg x_3 \oplus (a \land b) \\
    		x_0' & \gets b \oplus (c \land d) \\
    		e & \gets d \oplus x_0' \\
    		x_1' & \gets a \oplus d \\
    		x_2' & \gets c \oplus e \\
    		x_3' & \gets e \oplus (x_0' \land x_2')
    	\end{align*}
    	\caption{The 4-bit S-box $\textsf{Sb}_0$ of Midori64 computed using 4 AND gates.}
    	\label{fig:applications:spn:midori-sbox}
    \end{subfigure}
    \caption{Implementation formulas for the TWINE and Midori64 S-boxes. The input bits are $x_0$ through $x_3$, the output bits are  $x_0'$ through $x_3'$.}
    \label{fig:applications:spn:twine-midori-sbox}
\end{figure}

\section{Implementation of SPN Primitives}
\label{sec:appendix:ciphers}
In the following, we give a more detailed explanation of the implementation from Tables~\ref{tab:applications:spn:cipher-gate-counts} and~\ref{tab:applications:spn:ciphers} for each primitive. 

\subsection{AES}
The key schedule of AES-128 applies 4 S-boxes per round to the state. All remaining key schedule operations can be expressed using XOR gates. The AES S-box can be computed with 32 AND gates, as described by Boyar and Peralta~\cite{Boyar2010}. In the data path, 16 S-boxes are applied per round. The ShiftRows, MixColumns and AddRoundKey steps can be expressed with XOR gates. AES-128 defines 10 rounds.

For an implementation using projection gates, we first compose the key into 8-bit wires. Then, the key schedule can be computed by replacing the S-box with a single 8-bit projection gate computing the same functionality.
For the data path, we replace S-boxes with 8-bit projection gates. The mixing step in AES cannot be described with a binary matrix alone but we re-write the MixColumns step as
\[\scriptsize
    \begin{aligned}
    &\begin{pmatrix}
        2311\\
        1231\\
        1123\\
        3112
    \end{pmatrix}
    \begin{pmatrix}
        s_0 ~ s_4 ~ s_8 ~ s_{12} \\
        s_1 ~ s_5 ~ s_9 ~ s_{13} \\
        s_2 ~ s_6 ~ s_{10} ~ s_{14} \\
        s_3 ~ s_7 ~ s_{11} ~ s_{15}
    \end{pmatrix}
    =\\
    &\quad\begin{pmatrix}
        0111\\
        1011\\
        1101\\
        1110
    \end{pmatrix}
    \begin{pmatrix}
        s_0 ~ s_4 ~ s_8 ~ s_{12} \\
        s_1 ~ s_5 ~ s_9 ~ s_{13} \\
        s_2 ~ s_6 ~ s_{10} ~ s_{14} \\
        s_3 ~ s_7 ~ s_{11} ~ s_{15}
    \end{pmatrix}
    \oplus
    \begin{pmatrix}
        1100\\
        0110\\
        0011\\
        1001
    \end{pmatrix}
    \begin{pmatrix}
        f(s_0) ~ f(s_4) ~ f(s_8) ~ f(s_{12}) \\
        f(s_1) ~ f(s_5) ~ f(s_9) ~ f(s_{13}) \\
        f(s_2) ~ f(s_6) ~ f(s_{10}) ~ f(s_{14}) \\
        f(s_3) ~ f(s_7) ~ f(s_{11}) ~ f(s_{15})
    \end{pmatrix}
    \end{aligned}
\]
where $s_0, \dots s_{15}$ are the 8-bit cells of the state and $f(s) = 2s$ computes the finite field doubling in $\mathbb{GF}(2^8)$ defined for AES.
Therefore, we compute a round of AES with $2 \cdot 16$ 8-bit projection gates.
This yields a correct result, since $s \oplus f(s) = 3s$ in $\mathbb{GF}(2^8)$.

\subsection{CRAFT}
The key and tweak bits are first composed into 4-bit wires. The remaining key schedule is linear w.r.t. 4-bit wires.

The data path is linear except for the 16 S-boxes that are applied in each of the 30 rounds. CRAFT uses the Midori $\textsf{Sb}_0$ S-box which can be computed with 4 AND gates (see Fig.~\ref{fig:applications:spn:midori-sbox}), or one 4-bit projection gate.

\subsection{Fides}
The internal state of Fides is a $4 \times 8$ grid of 5-bit and 6-bit cells for Fides-80 and Fides-96, respectively. We can compute the 5-bit S-box with 10 AND gates (see Fig.~\ref{fig:appendix:fides-5-sbox}), or one 5-bit projection gate. The 6-bit S-box may be computed with 34 AND gates expressing each output bit in algebraic normal form. This approach doesn't aim to optimise the number of AND gates used. However, we count common terms from different output bits only once since they can be shared as intermediate results. In our garbling scheme, the S-box is expressed in one 6-bit projection gate.
\begin{figure}[H]
	\begin{equation*}
    \resizebox{\columnwidth}{!}{%
		$\begin{array}{rl}
			a & \gets x_0 \land x_2 \\
			b & \gets x_1 \land x_4 \\
			c & \gets x_2 \land x_3 \\
			d & \gets x_0 \land x_4 \\
			e & \gets x_2 \land x_4 \\
			f & \gets x_1 \land x_2 \\
		\end{array}
		\begin{array}{rl}
			& \\
			x_0' & \gets \neg (x_0 \oplus x_3 \oplus b \oplus a) \\
			x_1' & \gets x_4 \oplus c \oplus  d \oplus e \oplus (x_0 \land x_1) \\
			x_2' & \gets  x_3 \oplus x_4 \oplus a \oplus d \oplus f \oplus (x_3 \land x_4) \\
			x_3' & \gets x_1 \oplus x_4 \oplus a \oplus c \oplus f \oplus (x_1 \land x_3)   \\
			x_4' & \gets x_1 \oplus x_2 \oplus x_3 \oplus b \oplus e \oplus  f  \oplus (x_0 \land x_3)
		\end{array}
    $}
	\end{equation*}
	\caption{The 5-bit S-box of Fides~\cite{Bilgin2013} can be computed with 10 AND gates. Input bits are $x_0, \dots, x_4$, output bits are $x_0', \dots , x_4'$.}
	\label{fig:appendix:fides-5-sbox}
\end{figure}

\subsection{MANTIS}
The key $k=k_0||k_1$ is expanded as defined in~\cite{Beierle2016}:
\[
	k_0 || (k_0 >>> 1) \oplus (k_0 >> 63) || k_1\,.
\]
Afterwards we compose the required 4-bit wires for the expanded key costing 192 1-bit projection gates.
MANTIS uses the Midori $\textsf{Sb}_0$ S-box, which can be computed with 4 AND gates (see Fig.~\ref{fig:applications:spn:midori-sbox}), or one 4-bit projection gate.

\subsection{Midori64}
The key bits are first composed into 4-bit wires. The key schedule can then be computed using XOR gates between the 4-bit wires.

In the data path, all steps except for the S-box can be computed with XOR gates alone. The 4-bit S-box $\textsf{Sb}_0$ can be computed with 4 AND gates (see Fig.~\ref{fig:applications:spn:midori-sbox}), or one 4-bit projection gate.

\subsection{Piccolo}
The key schedule for Piccolo-80 and Piccolo-128 can be computed using only XOR gates after the key bits are composed to 4-bit wires.

Piccolo's data path applies the 16-bit function $F$ two times per round to half of the state. This function $F$ is composed of a parallel application of 4 4-bit S-boxes, followed by a mixing matrix multiplication, followed by another parallel application of 4 4-bit S-boxes.
\[\scriptsize
	F(s_0, s_1, s_2, s_3) = \mathcal{S} \left(
	\begin{pmatrix}
		2311\\
		1231\\
		1123\\
		3112
	\end{pmatrix}
	\mathcal{S}\left(
	\begin{pmatrix}
		s_0 \\s_1 \\s_2\\s_3
	\end{pmatrix}\right)
	\right)\enspace,
\]
where the function $\mathcal{S}$ applies the 4-bit S-box $S$ element-wise
\[\scriptsize
	\mathcal{S}\left(
	\begin{pmatrix}
		s_0 \\s_1 \\s_2\\s_3
	\end{pmatrix}\right)
	=
	\begin{pmatrix}
		S(s_0) \\S(s_1) \\S(s_2)\\S(s_3)
	\end{pmatrix}\,.
\]
The mixing matrix encodes multiplications with elements in the finite field $\mathbb{GF}(2^4)$ with the irreducible polynomial $x^4 + x +1$.
Clearly, Piccolo doesn't have the property of a binary mixing matrix. However, we can still provide an implementation with projection gates at additional cost.

We re-write the function $F$ as
\[\scriptsize
	F'(s_0, s_1, s_2, s_3) = \mathcal{S} \left(
		\begin{pmatrix}
			0111\\
			1011\\
			1101\\
			1110
		\end{pmatrix}
		\begin{pmatrix}
			f(s_0)\\f(s_1)\\f(s_2)\\f(s_3)
		\end{pmatrix}
		\oplus
		\begin{pmatrix}
			1100\\
			0110\\
			0011\\
			1001
		\end{pmatrix}
		\begin{pmatrix}
			g(s_0)\\g(s_1)\\g(s_2)\\g(s_3)
		\end{pmatrix}
	\right)
\]
where $f(s) = S(s)$ and $g(s) = 2S(s)$. Subsequently, we compute $f$, $g$ and the remaining S-box layer $\mathcal{S}$ via 4-bit projection gates. Therefore, $F'$ can be computed with $4+4+4 = 12$ 4-bit projection gates. This re-writing is correct because $f(s) \oplus g(s) = 3S(s)$ w.r.t $\mathbb{GF}(2^4)$.

\subsection{SKINNY}
The SKINNY cipher family comprises three tweakey (i.e., public tweak concatenated with secret key) sizes, 64, 128 and 192 bit, of which we include the
size 128-bit here. 
The key schedule for SKINNY-64-128 also includes the application of a linear feedback shift register (LFSR) to 8 per round. This LFSR is implemented with a 4-bit projection gate.

The SKINNY data path contains 16 4-bit S-boxes per round. Each S-box is implemented with 4 AND gates using the formula from~\cite{Beierle2016}, or one 4-bit projection gate.

\subsection{TWINE}
\label{sec:appendix:ciphers:twine}
The key bits are first composed into 4-bit wires. The key schedule is linear except for 2 and 3 S-box computations per round for TWINE-80 and TWINE-128, respectively. In total, the key schedule comprises 35 rounds with S-box computation for both TWINE-80 and TWINE-128.

The data path is the same for TWINE-80 and TWINE-128 and contains 8 S-boxes per round in 36 rounds. The S-box can be computed with 6 AND gates (see Fig.~\ref{fig:applications:spn:twine-sbox}), or one 4-bit projection gate.

\subsection{WAGE}
The internal state of the WAGE permutation is represented as 37 7-bit cells. We load the initial state by computing the 7-bit wire composition for all bits.

We write $s_i$ to denote the $i$-th 7-bit cell and $s_i'$ to denote the updated $i$-th 7-bit cell. The internal state is updated 111 times in the following procedure:
\[\scriptsize
	\begin{array}{ll}
		\mathsf{fb} \gets & \textsf{WGP}(s_{36}) \oplus s_{31} \oplus s_{30} \oplus s_{26} \oplus s_{24} \oplus s_{19} \oplus s_{13} \oplus s_{12} \\
        & \quad\oplus s_8 \oplus s_6 \oplus \textsf{Dbl}(s_0) \\
		s_5 \gets & s_5 \oplus \textsf{SB}(s_8) \\
		s_{11} \gets & s_{11} \oplus \textsf{SB}(s_{15}) \\
		s_{19} \gets & s_{19} \oplus \textsf{WGP}(s_{18}) \oplus rc_0 \\
		s_{24} \gets & s_{24} \oplus \textsf{SB}(s_{27}) \\
		s_{30} \gets & s_{30} \oplus \textsf{SB}(s_{34}) \\
		s_j' \gets & s_{j+1},\, 0 \le j \le 35 \\
		s_{36}' \gets & \mathsf{fb}\enspace.
    \end{array}
\]
The 7-bit functions \textsf{WGP}, \textsf{Dbl} and \textsf{SB} denote a Welch-Gong permutation, finite field doubling and a lightweight 7-bit S-box. All three are implemented using a 7-bit projection gate.
Further, $rc_0$ is a round-dependent constant.}

\bibliographystyle{IEEEtran}
\bibliography{IEEEabrv,bibliography}

\end{document}